\def\submission{0}

\ifnum\submission=0
\documentclass[11pt]{article}
\else
\documentclass{llncs}
\fi
\usepackage[utf8]{inputenc}
\usepackage{amsmath}
\usepackage{amsfonts}
\ifnum\submission=0
\usepackage{amsthm}
\fi
\usepackage{braket}
\usepackage{authblk}
\usepackage[colorlinks=true,allcolors=blue]{hyperref}
\ifnum\submission=0
\usepackage[margin=1in]{geometry}
\fi
\usepackage{xcolor}
\usepackage{algorithm}
\usepackage{mdframed} 
\mdfdefinestyle{figstyle}{ %
  linecolor=black!7, %
  backgroundcolor=black!7, %
  innertopmargin=10pt, %
  innerleftmargin=\textwidth, %
  innerrightmargin=\textwidth, %
  innerbottommargin=5pt %
}

\usepackage{comment}

\ifnum\submission=0
\newtheorem{theorem}{Theorem}[section]

\newtheorem{lemma}[theorem]{Lemma}

\newtheorem{remark}{Remark}

\fi
\newtheorem{fact}{Fact}

\def \sample { \overset{\hspace{0.1em}\mathsf{\scriptscriptstyle\$}}{\leftarrow} }

\newcommand{\ra}{\rightarrow}
\newcommand{\la}{\leftarrow}
\newcommand{\pro}{P}
\newcommand{\ver}{V}
\newcommand{\secpar}{n}
\newcommand{\negl}{\mathsf{negl}}
\newcommand{\lang}{L}
\newcommand{\key}{k}
\newcommand{\comy}{y}
\newcommand{\bfy}{\mathbf{y}}

\newcommand{\td}{\mathsf{td}}
\newcommand{\st}{\mathsf{st}}
\newcommand{\out}{\mathsf{out}}
\newcommand{\bit}{\{0,1\}}
\newcommand{\ans}{a}

\newcommand{\poly}{\mathsf{poly}}
\newcommand{\hil}{\mathcal{H}}

\newcommand{\defeq}{:=}

\newcommand{\Succ}{\mathsf{Succ}}

\newcommand{\A}{\mathcal{A}}
\newcommand{\B}{\mathcal{B}}

\newcommand{\ext}{\mathsf{Ext}}

\newcommand{\calX}{\mathcal{X}}
\newcommand{\calY}{\mathcal{Y}}
\newcommand{\calS}{\mathcal{S}}

\newcommand{\ot}{\otimes}

\newcommand{\func}{\mathsf{Func}}
\newcommand{\TT}{\mathtt{T}}
\newcommand{\PRG}{\mathsf{PRG}}

\newcommand{\QTM}{\mathsf{QTM}}
\newcommand{\QTIME}{\mathsf{QTIME}}

\newcommand{\Acc}{\mathsf{Acc}}
\newcommand{\RE}{\mathsf{RE}}
\newcommand{\crs}{\mathsf{crs}}
\newcommand{\ek}{\mathsf{ek}}
\newcommand{\rsetup}{\mathsf{RE}.\mathsf{Setup}}
\newcommand{\renc}{\mathsf{RE}.\mathsf{Enc}}
\newcommand{\rdec}{\mathsf{RE}.\mathsf{Dec}}

\newcommand{\inp}{\mathsf{inp}}
\newcommand{\Time}{\mathsf{Time}}
\newcommand{\Menc}{\widehat{M_\inp}}

\newcommand{\SNARK}{\mathsf{SNARK}}
\newcommand{\snark}{\mathsf{snark}}
\newcommand{\NP}{\mathsf{NP}}

\newcommand{\rela}{\mathcal{R}}

\newcommand{\FHE}{\mathsf{FHE}}
\newcommand{\fhe}{\mathsf{fhe}}
\newcommand{\fhekeygen}{\mathsf{FHE}.\mathsf{KeyGen}}
\newcommand{\fheenc}{\mathsf{FHE}.\mathsf{Enc}}
\newcommand{\fhedec}{\mathsf{FHE}.\mathsf{Dec}}
\newcommand{\fheeval}{\mathsf{FHE}.\mathsf{Eval}}
\newcommand{\sk}{\mathsf{sk}}
\newcommand{\pk}{\mathsf{pk}}
\newcommand{\ct}{\mathsf{ct}}

\newcommand{\setupeff}{\setup_{\mathsf{eff}}}
\newcommand{\vereff}{V_{\mathsf{eff}}}
\newcommand{\vereffone}{V_{\mathsf{eff},1}}
\newcommand{\vereffthree}{V_{\mathsf{eff},3}}
\newcommand{\vereffout}{V_{\mathsf{eff},\mathsf{out}}}
\newcommand{\proeff}{P_{\mathsf{eff}}}
\newcommand{\proefftwo}{P_{\mathsf{eff},2}}
\newcommand{\proefffour}{P_{\mathsf{eff},4}}

\newcommand{\setup}{\mathsf{Setup}}
\newcommand{\re}{\mathsf{re}}

\newcommand{\transcript}{\mathsf{trans}}

\newcommand{\setupefffs}{\setup_{\mathsf{eff}\text{-}\mathsf{fs}}}
\newcommand{\proefffs}{P_{\mathsf{eff}\text{-}\mathsf{fs}}}
\newcommand{\proefffstwo}{P_{\mathsf{eff}\text{-}\mathsf{fs},2}}
\newcommand{\verefffs}{V_{\mathsf{eff}\text{-}\mathsf{fs}}}
\newcommand{\verefffsone}{V_{\mathsf{eff}\text{-}\mathsf{fs},1}}
\newcommand{\verefffsout}{V_{\mathsf{eff}\text{-}\mathsf{fs},\out}}
\newcommand{\game}{\mathsf{Game}}

\newcommand*{\opro}[2]{|#1\rangle\langle#2|}
\newcommand*{\ipro}[2]{\langle #1|#2\rangle}
\newcommand{\TD}{\mathsf{TD}}

\newcommand*{\regK}{\mathbf{K}}

\newcommand*{\regX}{\mathbf{X}}
\newcommand*{\regY}{\mathbf{Y}}
\newcommand*{\regZ}{\mathbf{Z}}

\newcommand*{\regC}{\mathbf{C}}

\ifnum\submission=0
\newcommand{\nai}[1]{{\color{purple}[Nai: #1]}}
\newcommand{\km}[1]{{\color{brown}[KM: #1]}}
\newcommand{\takashi}[1]{{\color{orange}[Takashi: #1]}}
\else
\newcommand{\nai}[1]{}
\newcommand{\km}[1]{}
\newcommand{\takashi}[1]{}
\fi

\ifnum\submission=0
\title{Classical Verification of Quantum Computations with \\ Efficient Verifier}
\else
\title{Classical Verification of Quantum Computations with Efficient Verifier}
\fi
\begin{document}

\ifnum\submission=0
\author[1]{Nai-Hui Chia}
\author[2]{Kai-Min Chung}
\author[3]{Takashi Yamakawa\thanks{This work was done in part while the author was visiting Academia Sinica.}}
\affil[1]{Department of Computer Science, University of Texas at Austin}
\affil[2]{Institute of Information Science, Academia Sinica}
\affil[3]{NTT Secure Platform Laboratories}
\else
\author{\empty}\institute{\empty} 
\fi

\maketitle

\ifnum\submission=0
\vspace{-7mm} 
\fi

\begin{abstract}
In this paper, we extend the protocol of classical verification of quantum computations (CVQC) recently proposed by Mahadev to make the verification efficient.
Our result is obtained in the following three steps:
\begin{itemize}
    \item We show that parallel repetition of Mahadev's protocol has negligible soundness error. This gives the first constant round CVQC protocol with negligible soundness error. In this part, we only assume the quantum hardness of the learning with error (LWE) problem similarly to the Mahadev's work.
    \item We construct a two-round CVQC protocol in the quantum random oracle model (QROM) where a cryptographic hash function is idealized to be a random function.
    This is obtained by applying the Fiat-Shamir transform to the parallel repetition version of the Mahadev's protocol.
    \item We construct a two-round CVQC protocol with efficient verifier in the CRS+QRO model where both prover and verifier can access to a (classical) common reference string generated by a trusted third party in addition to quantum access to QRO.
    Specifically, the verifier can verify a $\QTIME(T)$ computation in time $\poly(\secpar,\log T)$ where $\secpar$ is the security parameter.
    For proving soundness, we assume that a standard model instantiation of our two-round protocol with a concrete hash function (say, SHA-3) is sound and the existence of post-quantum indistinguishability obfuscation and post-quantum fully homomorphic encryption in addition to the quantum hardness of the LWE problem. 
\end{itemize}
\end{abstract}
\section{Introduction}
Can we verify quantum computations by a classical computer? This problem has been a major open problem in the field until Mahadev~\cite{FOCS:Mahadev18a} finally gave an affirmative solution.
Specifically, she constructed an interactive protocol between an efficient classical verifier (a BPP machine) and an efficient quantum prover (a BQP machine) where the verifier can verify the result of the BQP computation.
(In the following, we call such a protocol a CVQC protocol.\footnote{``CVQC" stands for ``Classical Verification of Quantum Computations"}) 
Soundness of her protocol relies on a computational assumption that the learning with error (LWE) problem \cite{JACM:Regev09} is hard for an efficient quantum algorithm, which has been widely used in the field of cryptography. We refer to the extensive survey by Peikert \cite{FTTCS:Peikert16} for details about LWE and its cryptographic applications.

Though her result is a significant breakthrough, there are still several drawbacks. First, her protocol has soundness error $3/4$, which means that a cheating prover may convince the verifier even if it does not correctly computes the BQP computation with probability at most $3/4$. Though we can exponentially reduce the soundness error by sequential repetition, we need super-constant rounds to reduce the soundness error to be negligible.
If parallel repetition works to reduce the soundness error, then we need not increase the number of round.
However, parallel repetition may not reduce soundness error for computationally sound protocol in general \cite{FOCS:BelImpNao97,TCC:PieWik07}.
Thus, it is still open to construct constant round protocol with negligible soundness error.

Another issue is about verifier's efficiency. In her protocol, for verifying a computation that is done by a quantum computer in time $T$, the verifier's running time is as large as $\poly(T)$.
Considering a situation where a device with weak classical computational power outsources computations to untrusted quantum server, we may want to make the verifier's running time as small as possible.
Such a problem has been studied well in the setting where the prover is classical and we know solutions where verifier's running time only logarithmically depends on $T$~\cite{STOC:Kilian92,SIAM:Micali00,STOC:KalRazRot13,STOC:KalRazRot14,JACM:GolKalRot15,STOC:ReiRotRot16,STOC:BraHolKal17,STOC:BKKSW18,FOCS:HolRot18,STOC:CCHLRRW19,STOC:KalPanYan19}.
Hopefully, we want to obtain a CVQC protocol where the verifier runs in logarithmic time.  

\subsection{Our Results}
In this paper, we solve the above drawbacks of the Mahadev's protocol. 
Our contribution is divided into three parts:
\begin{itemize}
    \item We show that parallel repetition version of Mahadev's protocol has negligible soundness error. This gives the first constant round CVQC protocol with negligible soundness error.
    \item We construct a two-round CVQC protocol in the quantum random oracle model (QROM) \cite{AC:BDFLSZ11} where a cryptographic hash function is idealized to be a random function that is only accessible as a quantum oracle.
    This is obtained by applying the Fiat-Shamir transform \cite{C:FiaSha86,C:LiuZha19,C:DFMS19} to the parallel repetition version of the Mahadev's protocol.
    \item We construct a two-round CVQC protocol with logarithmic-time verifier in the CRS+QRO model where both prover and verifier can access to a (classical) common reference string generated by a trusted third party in addition to quantum access to QRO.
    For proving soundness, we assume that a standard model instantiation of our two-round protocol with a concrete hash function (say, SHA-3) is sound and the existence of post-quantum indistinguishability obfuscation \cite{JACM:BGIRSVY12,SIAM:GGH0SW16} and (post-quantum) fully homomorphic encryption (FHE) \cite{STOC:Gentry09} in addition to the quantum hardness of the LWE problem. 
\end{itemize}

\subsection{Technical Overview}   \label{subsec:tech-intro}
\paragraph{Overview of Mahadev's protocol.}
 First, we recall the high-level structure of Mahadev's 4-round CVQC protovol.\footnote{See Sec. \ref{sec:mahadev_overview} for more details.}
On input a common input $x$, a quantum prover and classical verifier proceeds as below to prove and verify that $x$ belongs to a BQP language $L$.
\begin{description}
\item[Fist Message:] The verifier generates a pair of ``key" $k$ and a ``trapdoor" $\td$, sends $k$ to the prover, and keeps $\td$ as its internal state.

\item[Second Message:] The prover is given the key $k$, generates a classical ``commitment" $y$ along with a quantum state $\ket{\st_{\pro}}$, sends $y$ to the verifier, and keeps $\ket{\st_{\pro}}$ as its internal state.

\item[Third Message:] The verifier randomly picks a ``challenge" $c\sample \bit$ and sends $c$ to the prover. Following the terminology in \cite{FOCS:Mahadev18a}, we call the case of $c=0$ the ``test round" and the case of $c=1$ the ``Hadamard round".

\item[Forth Message:]
The prover is given a challenge $c$, generates a classical ``answer" $a$ by using the state $\ket{\st_{\pro}}$, and sends $a$ to the verifier.

\item[Final Verification:]
Finally, the verifier returns $\top$ indicating acceptance or $\bot$ indicating rejection.
In case $c=0$, the verification can be done publicly, that is, the final verification algorithm need not use $\td$.
\end{description}

Mahadev showed that the protocol achieves negligible completeness error and constant soundness error against computationally bounded cheating provers. More precisely, she showed that if $x\in L$, then the verifier accepts with probability $1-\negl(n)$ where $n$ is the security parameter, and if $x\notin L$, then any quantum polynomial time cheating prover can let the verifier accept with probability at most $3/4$.
For proving this, she first showed the following lemma:\footnote{Strictly speaking, she just proved a similar property for what is called a ``measurement protocol" instead of CVQC protocol. But this easily implies a similar statement for CVQC protocol since CVQC protocol can be obtained by combining a measurement protocol and the (amplified version of) Morimae-Fitzsimons protocol \cite{PR:MorFit18} without affecting the soundness error as is done in \cite[Section 8]{FOCS:Mahadev18a}.}

\begin{lemma}[informal]\label{lem:Mah_soundness_informal}
For any $x\notin L$, if a  quantum polynomial time cheating prover passes the test round with probability $1-\negl(n)$, then it passes the Hadamard with probability $\negl(n)$ assuming the quantum hardness of the LWE problem.
\end{lemma}

Given the above lemma, it is easy to prove the soundness of the protocol. Roughly speaking, we consider a decomposition of the Hilbert space $\hil_{\pro}$ for the prover's internal state $\ket{\psi_{\pro}}$ into two subspaces $S_{0}$ and $S_1$ so that $S_0$ (resp. $S_1$) consists of quantum states that lead  to rejection (resp. acceptance) in the test round.
That is, we define these subspaces so that if the cheating prover's internal state after sending the second message is $\ket{s_0}\in S_0$ (resp. $\ket{s_1}\in S_1$), then the verifier returns rejection (acceptance) at last of the protocol. Here, we note that the decomposition is well-defined since we can assume that a cheating prover just applies a fixed unitary on its internal space and measure some registers for generating the forth message without loss of generality. 
Let $\Pi_b$ be the projection onto $S_b$ and $\ket{\psi_b}:= \Pi_{b}\ket{\psi_{\pro}}$ for $b\in \bit$.  
Then $\ket{\psi_{0}}$ leads to rejection in the test round (with probability $1$), so if the verifier uniformly chooses $c\sample \bit$, then $\ket{\psi_{0}}$ leads to acceptance with probability at most $1/2$. 
On the other hand, since $\ket{\psi_{1}}$ leads to the acceptance in the test round (with probability $1$), by Lemma~\ref{lem:Mah_soundness_informal},  $\ket{\psi_{1}}$ leads to the acceptance in the Hadamard round with only negligible probability. Therefore,  the verifier uniformly chooses $c\sample \bit$, then $\ket{\psi_{1}}$ leads to acceptance with probability at most $1/2+\negl(\secpar)$.
Therefore, intuitively speaking, $\ket{\psi_{\pro}}=\ket{\psi_0}+\ket{\psi_1}$ leads to acceptance with probability at most $1/2+\negl(\secpar)$, which completes the proof of soundness. \footnote{Here is a small gap since measurements are not linear and thus we cannot simply conclude that  $\ket{\psi_{\pro}}$ leads to acceptance with probability at most $1/2+\negl(\secpar)$ even though the same property holds for both $\ket{\psi_0}$ and $\ket{\psi_1}$. Indeed, Mahadev just showed that the soundness error is at most $3/4$ instead of $1/2+\negl(\secpar)$ to deal with this issue. A concurrent work by Alagic et al. \cite{arXiv:AlaChiHun19} proved that the Mahadev's protocol actually achieves soundness error $1/2+\negl(\secpar)$ with more careful analysis.}

\paragraph{Parallel repetition.}
Now, we turn our attention to parallel repetition version of the Mahadev's protocol. Our goal is to prove that the probability that the verifier accepts on $x\notin L$ is negligible if the verifier and prover run the Mahadev's protocol $m$-times parallelly for sufficiently large $m$ and the verifier accepts if and only if it accepts on the all coordinates.
Our first (failed) attempt is to apply a similar argument to the one for the stand-alone Mahadev's protocol. 
Specifically, we consider a decomposition of the Hilbert space $\hil_{\pro}$ into $S_{i,0}$ and $S_{i,1}$ for each $i\in[m]$ so that $S_{i,0}$ (resp. $S_{i,1}$) consists of quantum states that lead  to rejection (resp. acceptance) in the test round on the $i$-th coordinate.\footnote{As explained later, this sentence actually does not make sense.}
Then for a $c\in \bit^m$, we can decompose $\ket{\psi_{\pro}}$ as
\[
\ket{\psi_{\pro}}=\ket{\psi_{1}}+...+\ket{\psi_{m}}+\ket{\psi'}
\]
so that $\ket{\psi_{i}}\in S_{i, c_i}$ for $i\in [m]$ by recursively applying projections onto $S_{i,0}$ and $S_{i,1}$.
Then, $\ket{\psi_{i}}$ leads to acceptance on the $i$-th coordinate with a negligible probability if the challenge on the $i$-th coordinate is $c_i$ by using the same argument as in the stand-alone setting.  
Moreover, if we decompose it in this way for randomly chosen $c$, we have $E_{c\sample \bit^m}[\|\ket{\psi'}\|]\leq 2^{-m}$ since an expected norm is halved whenever we apply either of projections onto $S_{i,0}$ or $S_{i,1}$ randomly.
Overall, each term $\ket{\psi_{i}}$ leads to rejection on the $i$-th coordinate with overwhelming probability and the expected norm of the remaining term $\ket{\psi'}$ is negligible  if the verifier randomly picks a challenge $c\sample \bit^m$. 
Therefore it seems that we can conclude that the probability that the verifier accepts on the all coordinate can be only  negligible.

However, this does not work as it is. The problem is that a cheating prover's behavior in the forth round depends on challenges $c=c_1....c_m \in \bit^m$ on all coordinates.
Thus, even if we focus on analysis of the test round on the $i$-th round (i.e., the case of $c_i=0$), all other challenges $c_{-i}=c_1...c_{i-1}c_{i}...c_{m}$ still have flexibility, and a different choice of $c_{-i}$ leads to a different prover's behavior. In other words, the prover's strategy should be  described as a unitary over $\hil_{\regC} \ot \hil_{\pro}$ where $\hil_{\regC}$ is a Hilbert space to store a challenge.
Therefore $S_{i,0}$ and $S_{i,1}$ cannot be well-defined as a decomposition of $\hil_{\pro}$.

Here, we observe that what is essential to make the above approach work is to decompose $\hil_\pro$ into subspaces $S_{i,0}$ and $S_{i,1}$ so that the following property holds:
\ifnum\submission=0
\begin{quotation}
\emph{For $b\in \bit$, any state $\ket{s_{i,b}}\in S_{i,b}$ and any fixed $c\in \bit^m$ such that $c_i=b$,\\ $\ket{s_{i,b}}$ leads to acceptance on the $i$-th coordinate with ``small" probability.\footnote{If we could prove that the probability is negligible, then that would have made the rest of the proof simpler. However, we can only prove a slightly weaker statement that for any polynomial $p$, there exists a choice of parameters such that the probability is at most $1/p(\secpar)$. We use ``small" probability to mean this in this overview. We note that this is still sufficient for the rest of the proof.}}
\end{quotation}
\else
\begin{quotation}
\emph{For $b\in \bit$, any state $\ket{s_{i,b}}\in S_{i,b}$ and any fixed $c\in \bit^m$ such that $c_i=b$, $\ket{s_{i,b}}$ leads to acceptance on the $i$-th coordinate with ``small" probability.\footnote{If we could prove that the probability is negligible, then that would have made the rest of the proof simpler. However, we can only prove a slightly weaker statement that for any polynomial $p$, there exists a choice of parameters such that the probability is at most $1/p(\secpar)$. We use ``small" probability to mean this in this overview. We note that this is still sufficient for the rest of the proof.}}
\end{quotation}
\fi


Here, our idea is that instead of decomposing $\hil_\pro$ into ``always-reject" space and ``always-accept" space as above, we introduce a certain threshold $\gamma$, and we define $S_{i,0}$ (resp. $S_{i,1}$) as a subspace consisting of states that leads to acceptance in the test round on the $i$-th coordinate with probability smaller (resp. grater) than the threshold $\gamma$.

Then the case of $b=0$ is easy to analyze:
Since any state $s_{i,0}\in S_{i,0}$ leads to acceptance in the test round on the $i$-th coordinate with probability at most $\gamma$ when $c_{-i}$ is uniformly chosen from $\bit^{m-1}$. This implies that $s_{i,0}\in S_{i,0}$ leads to acceptance in the test round on the $i$-th coordinate with probability at most $2^{m-1}\gamma$ for any fixed $c_{-i}$.
Thus, we can show the first property if we set parameters so that $2^{m-1}\gamma$ is ``small".
\footnote{One may think that the exponential loss in $m$ is problematic since we consider super-logarithmic number of parallel repetition. However, by a standard technique, we can focus on the case of $m=O(\log n)$ without loss of generality. See the first paragraph of Sec. \ref{sec:proof_of_soundness} for details.}

The case of $b=1$ is a little trickier to argue.
First, we note that any state $\ket{s_{i,1}}\in S_{i,1}$ leads to acceptance in the test round on the $i$-th coordinate with probability at least $\gamma$ by definition. 
Then our idea is to amplify the probability to $1-\negl(n)$ by repeating the cheating prover's forth-message-generation algorithm $O(1/\gamma)$ times. 
If we can do this, then we would be able to show that $\ket{s_{i,1}}$ leads to acceptance in the Hadamard round with only a negligible probability by Lemma \ref{lem:Mah_soundness_informal}.
However, an obvious problem is that we cannot amplify the probability of passing the test round by simply repeating the cheating prover since the state $\ket{s_{i,1}}$ may be broken once the prover runs.
We overcome this problem by using the idea of \cite{MW05,NWZ09}, which give amplification theorem for QMA.
Very roughly speaking, they show that we can ``reuse" the same quantum state to amplify the probability to pass verification if the original state is an eigenvector of a certain operator associated with the verification procedure.
Based on this idea, we define $S_{i,0}$ and $S_{i,1}$ as spaces spanned by eigenvectors with eigenvalues smaller or grater than the threshold $\gamma$ so that we can amplify the probability to pass the verification when arguing the second property. 

 Now, we get close to our final goal, but there still remains an issue. 
 In the analysis for the case of $b=1$, we apply Lemma \ref{lem:Mah_soundness_informal} for a cheating prover that posseses $\ket{s_{i,1}}$, which is $S_{i,1}$ component of $\ket{\st_\pro}$ as its internal state.
 However, for applying this lemma, we have to ensure that we can efficiently generate $\ket{s_{i,1}}$ because the lemma can be applied only to efficient cheating provers.
 On the other hand, though we can assume that  $\ket{\st_\pro}$ can be efficiently generated, we do not know how to efficiently generate $\ket{s_{i,1}}$ since we do not know how to perform projection onto $S_{i,1}$.
 As explained above, we define $S_{i,1}$ as a space spanned by eigenvectors of a certain operator with eigenvalues larger than $\gamma$, so our primary idea is to use the phase estimation to estimate the eigenvalue as is done in \cite{NWZ09}.
 However, a problem is that we can only perform phase estimation with accuracy  $1/\poly(\secpar)$ in polynomial time in $\secpar$.
 Therefore, if the magnitude of $\ket{\st_\pro}$ on eigenvectors with eigenvalues within the range of $\gamma\pm 1/\poly(\secpar)$ is non-negligible, then the ``error" of the projection becomes non-negligible.
 Our final idea is to randomly chooses the threshold $\gamma$ from $T$ possible values.
 Then the expected magnitude on eigenvectors with eigenvalues within the range of $\gamma\pm 1/\poly(\secpar)$ is at most $1/T$ over the random choice of $\gamma$.
 By taking $T$ sufficiently large, we can finally resolve the problem and complete the soundness of parallel repetition version of the Mahadev's protocol. 
 
 It is worth noting that constructing the  efficient projection mentioned above is the most technically involved part in this work. We will provide further explanation before Lemma~\ref{lem:partition} in Section~\ref{sec:proof_of_soundness}. 
 
 \paragraph{Two-round protocol via Fiat-Shamir transform.}
 Here, we explain how to convert the parallel repetition version of the Mahadev's protocol to a two-round protocol in the QROM.
 First, we observe that the third message of the Mahadev's protocol is public-coin, and thus the parallel repetition version also satisfies this property. 
 Then by using the Fiat-Shamir transform  \cite{C:FiaSha86}, we can replace the third message with hash value of the transcript up to the second round.
 Though the Fiat-Shamir transform was originally proven sound only in the classical ROM, recent works \cite{C:LiuZha19,C:DFMS19} showed that it is also sound in the QROM.
 This enables us to apply the Fiat-Shamir transform to the parallel repetition version of the Mahadev's protocol to obtain a two-round protocol in the QROM.
 
 \paragraph{Making verification efficient.}
 Finally, we explain how to make the verification efficient.
 Our idea is to delegate the verification procedure itself to the prover by using delegation algorithm for classical computation.
 Since the verification is classical, this seems to work at first glance.
 However, there are the following two problems:
 \begin{enumerate}
      \item  There is not a succinct description of the verification procedure since the verification procedure is specified by the whole transcript whose size is $\poly(T)$ when verifying a language in $\QTIME(T)$. Then the verifier cannot specify the verification procedure to delegate within time $O(\log (T))$.
     \item Since the CVQC protocol is not publicly verifiable, the prover cannot know the description of the verification procedure, which is supposed to be delegated to the prover. 
 \end{enumerate}
 
  We solve the first problem by using the succinct randomized encoding, which enables one to generate a succinct encoding of a Turing machine $M$ and an input $x$ so that the encoding only reveals the information about $M(x)$ and not $M$ or $x$.
 Then our idea is that instead of sending the original first message, the verifier just sends a succinct encoding of $(V_1,s)$ where $V_1$ denotes the Turing machine that takes $s$ as input and works as the first-message-generation algorithm of the CVQC protocol with randomness $PRG(s)$ where $PRG$ is a pseudorandom number generator.
 This enables us to make the transcript of the protocol succinct (i.e., the description size is logarithmic in $T$) so that the verifier can specify the verification procedure succinctly. 
 To be more precise, we have to use the strong output-compressing randomized encoding \cite{AC:BFKSW19}, where the encoding size is independent of the output length of the Turing machine. 
 They construct the strong output-compressing randomized encoding based on iO and other mild assumptions in the common reference string.
 Therefore our CVQC protocol also needs the common reference string.
 
 We solve the second problem by using FHE.
 Namely, the verifier sends an encryption of the trapdoor $\td$ by FHE, and the prover performs the verification procedure over the ciphertext and provides a proof that it honestly applied the homomorphic evaluation by SNARK.
 Then the verifier decrypts the resulting FHE ciphertext and accepts if  the decryption result is ``accept" and the SNARK proof is valid.
 We note that it is recently shown that SNARK exists in the QROM \cite{TCC:ChiManSpo19}, but here is one thing we have to be careful about:
 They only proved the non-adaptive soundness for the SNARK, which only ensures the soundness in the setting where the statement is chosen before making any query to the random oracle.
 However, in the CVQC protocol, a cheating prover may choose a false statement while making queries to the random oracle. 
 To deal with this issue, we first expand the protocol to the four-round protocol where the verifier randomly sends a ``salt" $z$, which is a random string of a certain length, in the third round and the prover uses the ``salted" random oracle $H(z,\cdot)$ for generating the SNARK proof.
 Since the statement to be proven by SNARK is determined up to the second round, and the salting essentially makes the random oracle ``fresh", we can argue the soundness of the CVQC protocol even with the non-adaptive soundness of the SNARK.
 At this point, we obtain four-round CVQC protocol with efficient verification.
 Here, we observe that the third message is just a salt $z$, which is public-coin.
 Therefore we can just apply the Fiat-Shamir transform again to make the protocol two-round.

\subsection{Related Works}
\paragraph{Verification of Quantum Computation.}
There are long line of researches on verification of quantum computation.
Except for solutions relying on computational assumptions, there are two type of settings where verification of quantum computation is known to be possible.
In the first setting, instead of considering purely classical verifier, we assume that a verifier can perform a certain kind of weak quantum computations \cite{FOCS:BroFitKas09,PR:FitKas17,arXiv:ABOEM17,PR:MorFit18}.
In the second setting, we assume that a prover is splitted into two remote servers that share entanglement but do not communicate \cite{Nat:RUV13}.
Though these works do not give a CVQC protocol in our sense, the advantage is that we need not assume any computational assumption for the proof of soundness, and thus they are incomparable to Mahadev's result and ours.

Subsequent to Mahadev's breakthrough result, Gheorghiu and Vidick \cite{FOCS:GheVid19} gave a CVQC protocol that also satisfies blindness, which ensures that a prover cannot learn what computation is delegated.
We note that their protocol requires polynomial number of rounds.

\paragraph{Concurrent Work.}
In a concurrent and independent work, Alagic et al. \cite{arXiv:AlaChiHun19} also shows similar results to our first and second results, parallel repetition theorem for the Madadev's protocol and a two-round CVQC protocol by the Fiat-Shamir transform. We note that our third result, a two-round CVQC protocol with efficient verification, is unique in this paper. 

We mention that we have learned the problem of parallel repetition for Mahadev's protocol from the authors of~\cite{arXiv:AlaChiHun19} on March 2019, but investigated the problem independently later as a steppingstone toward making the verifier efficient. Interestingly, the analyses of parallel repetition in the two works are quite different. Briefly, the analysis in~\cite{arXiv:AlaChiHun19} relies on the observation that for any two different challenges $c_1 \neq c_2 \in \{0,1\}^m$, the projections of an efficient-generated prover's state on the accepting subspaces corresponding to $c_1$ and $c_2$ are almost orthogonal, which leads to an elegant proof of the parallel repetition theorem. On the other hand, the key step of our analysis is to construct an efficient projection that  (approximately) projects the prover's state to two states that will be rejected by the test and Hadamard round with high probability, respectively. We analyze the soundness error of the parallel repetition by applying the efficient projection to each coordinate sequentially. We believe that our construction of the efficient projection is a general technique  which may find other applications.

\section{Preliminaries}
\paragraph{Notations.}
For a bit $b\in\bit$, $\bar{b}$ denotes $1-b$.
For a finite set $\calX$, $x\sample \calX$ means that $x$ is uniformly chosen from $\calX$.
For finite sets $\calX$ and $\calY$, $\func(\calX,\calY)$ denotes the set of all functions with domain $\calX$ and range $\calY$.
A function $f:\mathbb{N}\ra [0,1]$ is said to be negligible if for all polynomial $p$ and sufficiently large $\secpar \in \mathbb{N}$, we have $f(\secpar)< 1/p(\secpar)$ and said to be overwhelming if $1-f$ is negligible.
We denote by $\poly$ an unspecified polynomial and by $\negl$ an unspecified negligible function.
We say that a classical (resp. quantum) algorithm is efficient if it runs in probabilistic polynomial-time (resp. quantum polynominal time).
For a quantum or randomized algorithm $\A$, $y\sample \A(x)$ means that $\A$ is run on input $x$ and outputs $y$ and $y \defeq \A(x;r)$  means that $\A$ is run on input $x$ and randomness $r$ and outputs $y$. 
For an interactive protocol between a ``prover" $P$ and ``verifier" $V$, $y\sample \langle P(x_P),V(x_V)) \rangle(x_P)$ means an interaction between them with prover's private input $x_P$ verifier's private input $x_V$, and common input $x$ outputs $y$.
For a quantum state $\ket{\psi}$, $M_{\regX}\circ \ket{\psi}$ means a measurement in the computational basis on the register $\regX$ of $\ket{\psi}$.
We denote by $\QTIME(T)$ a class of languages decided by a quantum algorithm whose running time is at most $T$.
We use $\secpar$ to denote the security parameter throughout the paper.

\subsection{Learning with Error Problem}
Roughly speaking, the learning with error (LWE) is a problem to solve system of noisy linear equations.
Regev \cite{JACM:Regev09} proved that the hardness of LWE can be reduced to hardness of certain worst-case lattice problems via quantum reductions.
We do not give a definition of LWE in this paper since we use the hardness of LWE only for ensuring the soundness of the Mahadev's protocol (Lemma~\ref{lem:Mah_soundness}), which is used as a black-box manner in the rest of the paper.
Therefore, we use exactly the same assumption as that used in \cite{FOCS:Mahadev18a}, to which we refer for detailed definitions and parameter settings for LWE.

\subsection{Quantum Random Oracle Model}
The quantum random oracle model (QROM) \cite{AC:BDFLSZ11} is an idealized model where a real-world hash function is modeled as a quantum oracle that computes a random function.
More precisely, in the QROM, a random function $H:\calX \ra \cal Y$ of a certain domain $\calX$ and range $\calY$ is uniformly chosen from $\func(\calX,\calY)$ at the beginning, and every party (including an adversary) can access to a quantum oracle $O_H$ that maps $\ket{x}\ket{y}$ to $\ket{x}\ket{y\oplus H(x)}$.
We often abuse notation to denote $\A^{H}$ to mean a quantum algorithm $\A$ is given oracle $O_H$.  

\subsection{Cryptographic Primitives}
\ifnum\submission=1
We give definitions of cryptographic primitives that are used in this paper in Appendix \ref{app:cryptofgraphic_primitives}.
We note that they are only used in Sec~\ref{sec:efficient} where we construct an efficient verifier variant.
\else
Here, we give definitions of cryptographic primitives that are used in this paper.
We note that they are only used in Sec~\ref{sec:efficient} where we construct an efficient verifier variant.

\subsubsection{Pseudorandom Generator}
A post-quantum pseudorandom generator (PRG) is an efficient deterministic classical algorithm $\PRG:\bit^{\ell} \ra \bit^{m}$ such that for any efficient quantum algorithm $\A$, we have
\begin{align*}
    \left|\Pr_{s\sample \bit^{\ell}}[\A(\PRG(s))]-\Pr_{y\sample \bit^{m}}[\A(y)]\right|\leq \negl(\secpar).
\end{align*}

It is known that there exists a post-quantum PRG for any $\ell=\Omega(\secpar)$ and $m=\poly(\secpar)$ assuming post-quantum one-way function \cite{SIAM:HILL99,FOCS:Zhandry12}.
Especially, a post-quantum PRG exists assuming the quantum hardness of LWE.



\subsubsection{Fully Homomorphic Encryption}
A post-quantum fully homomorphic encryption consists of four efficient classical algorithm $\Pi_\FHE=\allowbreak (\fhekeygen,\fheenc,\fheeval,\fhedec)$.
\begin{description}
\item[$\fhekeygen(1^\secpar)$:] The key generation algorithm takes the security parameter $1^\secpar$ as input and outputs a public key $\pk$ and a secret key $\sk$.
\item[$\fheenc(\pk,m)$:] The encryption algorithm takes a public key $\pk$ and a message $m$ as input, and outputs a ciphertext $\ct$.
\item[$\fheeval(\pk,C,\ct)$:] The evaluation algorithm takes a public key $\pk$, a classical circuit $C$, and a ciphertext $\ct$, and outputs a evaluated ciphertext $\ct'$.
\item[$\fhedec(\sk,\ct)$:] The decryption algorithm takes secret key $\sk$ and a ciphertext $\ct$ as input and outputs a message $m$ or $\bot$. 
\end{description}
\paragraph{Correctness.}
For all $\secpar\in \mathbb{N}$, $(\pk,\sk)\sample \fhekeygen(1^\secpar)$, $m$ and $C$, we have
\begin{align*}
    \Pr[\fhedec(\sk,\fheenc(\pk,m))=m]=1
\end{align*}
and
\begin{align*}
    \Pr[\fhedec(\sk,\fheeval(\pk,C,\fheenc(\pk,m)))=C(m)]=1.
\end{align*}
\paragraph{Post-Quantum CPA-Security.}
For any efficient quantum adversary $\A=(\A_1,\A_2)$, we have  
\begin{align*}
    &|\Pr[1\sample  \A_2(\ket{\st_{\A}},\ct):(\pk,\sk)\sample \fhekeygen(1^\secpar),(m_0,m_1,\ket{\st_{\A}})\sample \A_1(\pk), \ct\sample \fheenc(\pk,m_0)]\\
    &-\Pr[1\sample  \A_2(\ket{\st_{\A}},\ct):(\pk,\sk)\sample \fhekeygen(1^\secpar),(m_0,m_1,\ket{\st_{\A}})\sample \A_1(\pk), \ct\sample \fheenc(\pk,m_1)]|\\
&\leq \negl(\secpar).
\end{align*}

FHE is usually constructed by first constructing \textit{leveled} FHE, where we have to upper bound the depth of a circuit to evaluate at the setup, and then converting it to FHE by the technique called bootstrapping~\cite{STOC:Gentry09}. 
There have been many constructions of leveled FHE whose (post-quantum) security can be reduced to the (quantum) hardness of LWE \cite{FOCS:BraVai11,ITCS:BraGenVai12,C:Brakerski12,C:GenSahWat13}.
FHE can be obtained assuming that any of these schemes is \textit{circular secure} \cite{EC:CamLys01} so that it can be upgraded into FHE via bootstrapping.
We note that Canetti et al. \cite{TCC:CLTV15} gave an alternative transformation from leveled FHE to FHE based on subexponentially secure iO.

\subsubsection{Strong Output-Compressing Randomized Encoding}
A strong output-compressing randomized encoding \cite{AC:BFKSW19} consists of three efficient classical  algorithms $(\rsetup,\renc,\rdec)$.
\begin{description}
\item[$\rsetup(1^\secpar,1^\ell,\crs)$]: It takes the security parameter $1^\secpar$, output-bound $\ell$, and a common reference string $\crs\in\bit^{\ell}$ and outputs a encoding key $\ek$. 
\item[$\renc(\ek,M,\inp,T)$:] It takes an encoding key $\ek$, Turing machine $M$, an input $\inp\in \bit^*$, and a time-bound $T\leq 2^{\secpar}$ (in binary) as input and outputs an encoding $\Menc$.
\item[$\rdec(\crs,\Menc)$:] It takes a common reference string $\crs$ and an encoding $\Menc$ as input and outputs $\out \in \bit^* \cup \{\bot\}$.
\end{description}

\paragraph{Correctness.}
For any $\secpar\in\mathbb{N}$, $\ell,T\in \mathbb{N}$, $\crs\in \bit^{\ell}$, Turing machine $M$ and input $\inp\in \bit^*$ such that $M(\inp)$ halts in at most $T$ steps and returns a string whose length is at most $\ell$, we have 
\begin{align*}
    \Pr\left[\rdec(\Menc,\crs)=M(\inp): \ek\sample\rsetup(1^\secpar,1^\ell,\crs), \Menc\sample\renc(\ek,M,\inp,T)\right]=1.
\end{align*}

\paragraph{Efficiency.}
There exists polynomials $p_1,p_2,p_3$ such that for all $\secpar\in \mathbb{N}$, $\ell\in \mathbb{N}$, $\crs\sample \bit^{\ell}$:
\begin{itemize}
    \item If $\ek\sample \rsetup(1^\secpar,1^\ell,\crs)$, $|\ek|\leq p_1(\secpar, \log \ell)$.
    \item For every Turing machine $M$, time bound $T$, input $\inp\in \bit^{*}$, if  $\Menc\sample\renc(\ek,M,\inp,T)$, then $|\Menc|\leq p_2(|M|,|\inp|,\log T, \log \ell, \secpar)$,
    \item The running time of $\rdec(\crs,\Menc)$ is at most $\min(T,\Time(M,x))\cdot p_3(\secpar,\log T)$
\end{itemize}

\paragraph{Post-Quantum Security.}

There exists a simulator $\calS$ such that for any $M$ and $\inp$ such that $M(\inp)$ halts in $T^*\leq T$ steps and $|M(\inp)|\leq \ell$ and efficient quantum adversary $\A$,
\begin{align*}
    &|\Pr[1 \sample \A(\crs,\ek,\Menc):\crs\sample \bit^{\ell},\ek\sample\rsetup(1^\secpar,1^\ell,\crs),\Menc\sample\renc(\ek,M,\inp,T)]\\
    &-\Pr[1 \sample \A(\crs,\ek,\Menc):(\crs,\Menc)\sample \calS(1^{\secpar},1^{|M|},1^{|\inp|},M(\inp),T^*),\ek\sample\rsetup(1^\secpar,1^\ell,\crs)]|\leq \negl(\secpar).
\end{align*}

Badrinarayanan et al.~\cite{AC:BFKSW19} gave a construction of strong output-compressing randomized encoding based on iO and the LWE assumption. 

\subsubsection{SNARK in the QROM}
Let $H:\bit^{2\secpar}\ra \bit^{\secpar}$ be a quantum random oracle.
A SNARK for an $\NP$ language $\lang$ associated with a relation $\rela$ in the QROM consists of two efficient oracle-aided classical algorithms $\pro_\snark^H$ and $\ver_\snark^H$.
\begin{description}
\item[$\pro_\snark^H$:] It is an instance $x$ and a witness $w$ as input and outputs a proof $\pi$.
\item[$\ver_\snark^H$:] It is an instance $x$ and a proof $\pi$ as input and outputs $\top$ indicating acceptance or $\bot$ indicating rejection.
\end{description}
We require SNARK to satisfy the following properties:

\noindent\textbf{Completeness.}
For any $(x,w)\in\rela$, we have 
\begin{align*}
    \Pr_{H}[\ver_\snark^H(x,\pi)=\top:\pi\sample \pro_\snark^H(x,w)]=1.
\end{align*}

\noindent\textbf{Extractability.}
There exists an efficient quantum extractor $\ext$ such that
for any $x$ and a malicious quantum prover $\tilde{\pro}_\snark^{H}$ making at most $q=\poly(\secpar)$ queries, if 
\begin{align*}
    \Pr_H[\ver_\snark^H(x,\pi):\pi \sample \tilde{\pro}_\snark^H(x)]
\end{align*}
is non-negligible in $\secpar$, then 
\begin{align*}
    \Pr_H[(x,w)\in \rela: w\sample \ext^{\tilde{\pro}_\snark}(x,1^q,1^\secpar)]
\end{align*}
is non-negligible in $\secpar$.\\

\noindent\textbf{Efficient Verification.}
If we can verify that $(x,w)\in \rela$ in classical time $T$, then for any $\pi \sample \tilde{\pro}_\snark^H(x)$, $\ver_\snark^H(x,\pi)$ runs in classical time $\poly(\secpar,|x|,\log T)$.

Chiesa et al.~\cite{TCC:ChiManSpo19} showed that there exists SNARK in the QROM that satisfies the above properties.
\fi

\subsection{Lemma}
Here, we give a simple lemma, which is used in the proof of soundness of parallel repetition version of the Mahadev's protocol in Sec. \ref{sec:proof_of_soundness}.
\begin{lemma}\label{lem:Cauchy-Schwarz}
Let $\ket{\psi}=\sum_{i=1}^{m} \ket{\psi_i}$ be a quantum state and $M$ be a projective measurement.
Then we have
\begin{align*}
    \Pr[M\circ \ket{\psi}=1]\leq m \sum_{i=1}^{m}\|\ket{\psi_i}\|^2 \Pr[M\circ \frac{\ket{\psi_i}}{\|\ket{\psi_i}\|}=1]
\end{align*}
\end{lemma}
\ifnum\submission=1
The proof can be found in Appendix \ref{app:proof_CS}
\else
\begin{proof}
Since $M$ is a projective measurement, there exists a projection $\Pi$ such that 
\begin{align*}
    \Pr[M\circ \ket{\psi}=1]=\bra{\psi} \Pi \ket{\psi}.
\end{align*}
Then we have 
\begin{align*}
    \bra{\psi} \Pi \ket{\psi}&
    =\|\sum_{i=1}^{m} \Pi\ket{\psi_i}\|^2\\
    &\leq m \sum_{i=1}^{m}\|\Pi\ket{\psi_i}\|^2\\
    &=m\sum_{i=1}^{m}\bra{\psi_i} \Pi \ket{\psi_i}\\
    &=m \sum_{i=1}^{m}\|\ket{\psi_i}\|^2 \Pr[M\circ \frac{\ket{\psi_i}}{\|\ket{\psi_i}\|}=1]
\end{align*}
where we used the Cauchy-Schwarz inequality from the second to third lines.
\end{proof}
\fi

\section{Parallel Repetition of Mahadev's Protocol}

\subsection{Overview of Mahadev's Protocol}\label{sec:mahadev_overview}
Here, we recall the Mahadev's protocol \cite{FOCS:Mahadev18a}. We only give a high-level description of the protocol and properties of it and omit the details since they are not needed to show our result. 

The protocol is run between a quantum prover $\pro$ and a classical verifier $\ver$ on a common input $x$. The aim of the protocol is to enable a verifier to classically verify $x\in \lang$ for a BQP language $\lang$ with the help of interactions with a quantum prover.
The protocol is a 4-round protocol where the first message is sent from $\ver$ to $\pro$. 
We denote the $i$-th message generation algorithm by $\ver_i$ for $i\in\{1,3\}$ or $\pro_i$ for $i\in \{2,4\}$ and denote the verifier's final decision algorithm by $\ver_\out$.
Then a high-level description of the protocol is given below.
\begin{description}
\item[$\ver_1$:] On input the security parameter $1^\secpar$ and $x$, it generates a pair $(\key,\td)$ of a``key" and ``trapdoor", sends $\key$ to $\pro$, and keeps $\td$ as its internal state.
\item[$\pro_2$:] On input $x$ and $\key$, it generates a classical ``commitment" $\comy$ along with a quantum state $\ket{\st_\pro}$, sends $\comy$ to $\pro$, and keeps $\ket{\st_\pro}$ as its internal state.
\item[$\ver_3$:] It randomly picks a ``challenge" $c\sample \bit$ and sends $c$ to $\pro$.\footnote{The third message is just a public-coin, and does not depend on the transcript so far or $x$.}
Following the terminology in \cite{FOCS:Mahadev18a}, we call the case of $c=0$ the ``test round" and the case of $c=1$ the ``Hadamard round".
\item[$\pro_4$:] On input $\ket{\st_\pro}$ and $c$, it generates a classical ``answer" $\ans$ and sends $\ans$ to $\pro$.
\item[$\ver_\out$:] On input $\key$, $\td$, $y$, $c$, and $\ans$, it returns $\top$ indicating acceptance or $\bot$ indicating rejection.
In case $c=0$, the verification can be done publicly, that is, $\ver_\out$ need not take $\td$ as input.
\end{description}

For the protocol, we have the following properties:\\
\noindent\textbf{Completeness:}
For all $x\in \lang$, we have $\Pr[\langle \pro,\ver \rangle(x)]=\bot]= \negl(\secpar)$.\\
\noindent\textbf{Soundness:}
If the LWE problem is hard for quantum polynomial-time algorithms, then for any $x\notin \lang$ and a quantum polynomial-time cheating prover $\pro^*$, we have  $\Pr[\langle \pro^*,\ver \rangle(x)]=\bot]\leq 3/4$.

We need a slightly different form of soundness implicitly shown in \cite{FOCS:Mahadev18a}, which roughly says that if a cheating prover can pass the ``test round" (i.e., the case of $c=0$) with overwhelming probability, then it can pass the ``Hadamard round" (i.e., the case of $c=1$) only with a negligible probability. 
\begin{lemma}[implicit in \cite{FOCS:Mahadev18a}]\label{lem:Mah_soundness}
If the LWE problem is hard for quantum polynomial-time algorithms, then for any $x\notin \lang$ and a quantum polynomial-time cheating prover $\pro^*$ such that  $\Pr[\langle \pro^*,\ver \rangle(x)]=\bot\mid c=0]=\negl(\secpar)$, we have $\Pr[\langle \pro^*,\ver \rangle(x)]=\top\mid c=1]=\negl(\secpar)$.
\end{lemma}

We will also use the following simple fact:
\begin{fact}\label{fact:perfectly_pass_test}
There exists an efficient prover that passes the test round with probability $1$ (but passes the Hadamard round with probability $0$) even if $x\notin \lang$. 
\end{fact}

\subsection{Parallel Repetition}
Here, we prove that the parallel repetition of the Mahadev's protocol decrease the soundness bound to be negligible.
Let $\pro^m$ and $\ver^m$ be $m$-parallel repetitions of the honest prover $\pro$ and verifier $\ver$ in the Mahadev's protocol. Then we have the following:
\begin{theorem}[Completeness]\label{thm:rep_completeness}
For all $m= \Omega(\log^2(\secpar))$, for all $x\in \lang$, we have $\Pr[\langle \pro^m,\ver^m \rangle(x)]=\bot]= \negl(\secpar)$.\\
\end{theorem}
\begin{theorem}[Soundness]\label{thm:rep_soundness}
For all $m= \Omega(\log^2(\secpar))$, if the LWE problem is hard for quantum polynomial-time algorithms, then for any $x\notin \lang$ and a quantum polynomial-time cheating prover $\pro^*$, we have  $\Pr[\langle \pro^*,\ver^m \rangle(x)]=\top]\leq \negl(\secpar)$.
\end{theorem}
The completeness (Theorem~\ref{thm:rep_completeness}) easily follows from the completeness of the Mahadev's protocol.
In the next subsection, we prove the soundness (Theorem~\ref{thm:rep_soundness}).

\subsection{Proof of Soundness}\label{sec:proof_of_soundness}
First, we remark that it suffices to show that for any $\mu=1/\poly(n)$, there exists $m=O(\log(n))$ such that the success probability of the cheating prover is at most $\mu$.
This is because we are considering $\omega(\log(n))$-parallel repetition, in which case the number of repetitions is larger than   any $m=O(\log(n))$ for sufficiently large $n$, and thus we can just focus on  the first $m$ coordinates ignoring the rest of the coordinates.  
Thus, we prove the above claim in this section.

\noindent\textbf{Characterization of cheating prover.}
Any cheating prover can be characterized by a tuple $(U_0,U)$ of unitaries over Hilbert space $\hil_{\regC}\otimes \hil_{\regX}\otimes \hil_{\regZ} \otimes \hil_{\regY}  \otimes \hil_{\regK}$. 
A prover characterized by $(U_0,U)$ works as follows.\footnote{Here, we hardwire into the cheating prover the instance $x\notin \lang$ on which it will cheat instead of giving it as an input.}
\begin{description}
\item[Second Message:] Upon receiving $k=(\key_1,...,\key_m)$, it applies $U_0$ to the state $\ket{0}_{\regX}\otimes\ket{0}_{\regZ}\otimes\ket{0}_{\regY}\otimes  \ket{k}_{\regK}$, and then measures the $Y$ register to obtain $y=(\comy_1,...,\comy_m)$. Then it sends $\bfy$ to $\ver$ and keeps the resulting state $\ket{\psi(k,y)}_{\regX,\regZ}$ over  $\hil_{\regX,\regZ}$.
\item [Forth Message:] Upon receiving $c\in \bit^{m}$, it applies $U$ to $\ket{c}_{\regC}\ket{\psi(k,y)}_{\regX,\regZ}$ and then measures the $\regX$ register in computational basis to obtain $a=(a_1,...,a_m)$. We denote the designated register for $a_i$ by $\regX_i$. 
\end{description}

Here, we first introduce a general lemma about two projectors that was shown by Nagaj, Wocjan, and Zhang \cite{NWZ09} by using the Jordan's lemma.
\begin{lemma}[{\cite[Appendix A]{NWZ09}}]\label{lem:decomposition}
Let $\Pi_0$ and $\Pi_1$ be projectors on an $N$-dimensional Hilbert space $\hil$ and let $R_0:= 2\Pi_0-I$, $R_1:= 2\Pi_1-I$, and $Q:=R_1R_0$. $\hil$ can be decomposed into two-dimensional subspaces $S_j$ for $j\in[\ell]$ and $N-2\ell$ one-dimensional subspaces $T_j^{(bc)}$ for $b,c\in \bit$ that satisfies the following properties:
\begin{enumerate}
\item For each two-dimensional subspace $S_j$, there exist two orthonormal bases $(\ket{\alpha_j},\ket{\alpha_j^{\bot}})$ and $(\ket{\beta_j},\ket{\beta_j^{\bot}})$ of $S_j$ such that $\ipro{\alpha_j}{\beta_j}$ is a positive real and for all $ \ket{s}\in S_j$,  $\Pi_0 \ket{s} = \ipro{\alpha_j}{s}\ket{\alpha_j}$ and $\Pi_1 |s\rangle =\ipro{\beta_j}{s}\ket{\beta_j}$.  Moreover, $Q$ is a rotation with eigenvalues $e^{\pm i\theta_j}$ in $S_j$ corresponding to eigenvectors $\ket{\phi_j^+}=\frac{1}{\sqrt{2}}(\ket{\alpha_j}+i\ket{\alpha_j^\bot})$ and $\ket{\phi_j^-}=\frac{1}{\sqrt{2}}(\ket{\alpha_j}-i\ket{\alpha_j^\bot})$
where $\theta_j=2\arccos \ipro{\alpha_j}{\beta_j}=2\arccos\sqrt{\bra{\alpha_j}\Pi_1\ket{\alpha_j}}$.
\item  Each one-dimensional subspace $T_j^{(bc)}$ is spanned by a vector $\ket{\alpha_j^{(bc)}}$ such that $\Pi_0 \ket{\alpha_j^{(bc)}}=b \ket{\alpha_j^{(bc)}}$ and $\Pi_1 \ket{\alpha_j^{(bc)}}=c \ket{\alpha_j^{(bc)}}$.
\end{enumerate}
\end{lemma}

We first give ideas about Lemma~\ref{lem:partition} that is the main lemma for this section. For each coordinate $i\in[m]$, we would like to separate the space $\hil_{\regX,\regZ}$ into subspace that ``can pass the test round'' ($S_{i,1}$) and the subspace that ``cannot pass the test round'' ($S_{i,0}$). The first attempt is applying Lemma~\ref{lem:decomposition} to separate the space $\hil_{\regX,\regZ}$ into $S_{i,0}$ and $S_{i,1}$ such that the prover can make the verifier accept by using the amplification algorithm in~\cite{MW05} if the prover's state is in $S_{i,1}$. Then, the same prover will be rejected with $1-\negl(n)$ probability by Lemma~\ref{lem:Mah_soundness}. 

However, two things are making this approach fail. First, since the prover's unitary also operates on $\regC_{-i}$, we also need to consider $\regC_{-i}$ in addition to $\regX,\regZ$. Second, to have the binding property in Lemma~\ref{lem:Mah_soundness}, the prover's inner state must be constructed efficiently, which is not guaranteed by Lemma~\ref{lem:decomposition}. We bypass these difficulties by assuming the state in $\regC_{-i}$ is a uniformly superposed state and constructing an efficient procedure in Procedure~\ref{fig:process_G} to separate the space $\hil_{\regX,\regZ}$. Assuming the state in $\regC_{-i}$ is a uniformly superposed state will time the prover's success probability by $2^{(m-1)}$, which is still small if $m = O(\log n)$. 

Fix $k$ and $y$ for now (until we finish the proof of Lemma~\ref{lem:partition_further}). For $i\in [m]$, we consider two projectors 
\begin{align*}
    &\Pi_{in}:= \opro{0^m}{0^m}_{\regC}\otimes I_{\regX,\regZ}\\
   & \Pi_{i,out} := (UH_{\regC_{-i}})^{\dag}(\sum_{a_i\in \Acc_{k_i,y_i}}\opro{a_i}{a_i}_{\regX_i}\otimes I_{\regC,\regX_{-i},\regZ}) (UH_{\regC_{-i}}),
\end{align*}
where 
$\regX_{-i}:= \regX_1,\dots,\regX_{i-1}, \regX_{i+1},\dots, \regX_{m}$,  $H_{\regC_{-i}}$ means applying Hadamard operators to registers $\regC_1,\dots,\regC_{i-1}, \regC_{i+1},\dots, \regC_{m}$ to produce uniformly random challenges, and $\Acc_{k_i,y_i}$ denotes the set of $a_i$ such that the verifier accepts $a_i$ in the test round on the $i$-th coordinate when the first and second messages are $k_i$ and $y_i$, respectively.
Note that one can efficiently check if $a_i\in \Acc_{k_i,y_i}$ without knowing the trapdoor behind $k_i$ since verification in the test round can be done publicly as explained in Sec. \ref{sec:mahadev_overview}. 
We apply Lemma~\ref{lem:decomposition} for $\Pi_0= \Pi_{in}$ and $\Pi_1=\Pi_{i,out}$ to decompose the space  $\hil_{\regC,\regX,\regZ}$ into the two-dimensional subspaces $\{S_j\}_{j}$ and one-dimensional subspaces $\{T_{j}^{bc}\}_{j,b,c}$.
In the following, we use notations defined in Lemma~\ref{lem:decomposition} for this particular application.
We can write $\ket{\alpha_j}_{\regC,\regX,\regZ}=\ket{0}\ot \ket{\hat{\alpha}_j}_{\regX,\regZ}$  since $\Pi_{in}\ket{\alpha_j}=\ipro{\alpha_j}{\alpha_j}\ket{\alpha_j}=\ket{\alpha_j}$.
Similarly, we can write  $\ket{\alpha_j^{(10)}}_{\regC,\regX,\regZ}=\ket{0}\ot \ket{\hat{\alpha}_j^{(10)}}_{\regX,\regZ}$ 
and  $\ket{\alpha_j^{(11)}}_{\regC,\regX,\regZ}=\ket{0}\ot \ket{\hat{\alpha}_j^{(11)}}_{\regX,\regZ}$. 
Then  $\{\ket{\hat{\alpha}_j}\}_{j}$ and $\{\ket{\hat{\alpha}_{j}^{(1c)}}\}_{j,c}$ span $\hil_{\regX,\regZ}$. 

We let $S_{i,1}$ be the subspace spanned by eigenvectors $\ket{\hat{\alpha}}_{\regX,\regZ}$ which have large eigenvalues (corresponding to high accepting probability) and $S_{i,0}$ be the subspace spanned by eigenvectors $\ket{\hat{\alpha}}_{\regX,\regZ}$ which have small eigenvalues (corresponding to low accepting probability). For the efficient procedure in Procedure~\ref{fig:process_G}, one possible approach is using the amplification algorithm by Marriotts and Watrous~\cite{MW05}. Then, we let $S_{i,0}$ be the subspace accepted by the verifier in the test round with probability $1-\negl(n)$ after the amplification. However, this approach does not work since there will be a big subspace where the verifier accepts with non-negligible probability but not $1-\negl(n)$ probability, i.e., not in either $S_{i,0}$ nor $S_{i,1}$. It is hard to analyze the prover's success probability when its state is in this subspace. To overcome this difficulty, we apply the phase estimation to the operator $Q$ given in Lemma~\ref{lem:decomposition}, set a threshold $\gamma$ for eigenvalues, and undo the phase estimation. Then, we let the subspace spanned by eigenvectors with eigenvalues greater than $\gamma$ be $S_{i,1}$ and the subspace spanned by the others be $S_{i,0}$. Note that this approach still gives the error from the accuracy of phase estimation. For instance, let $\delta$ be the accuracy. Then the subspace with eigenvalues in $[\gamma-\delta,\gamma]$ will be the subspace which cannot be in either $S_{i,0}$ nor $S_{i,1}$. We handle this by choosing a random threshold $\gamma$ from $T$ thresholds such that the average size of the ``grey subspace'' will be $\leq 1/T$. In this way, we can get an approximate efficient projector $G_i$ such that it separates the space of $\regX,\regZ$ into three subspaces: $S_{i,0}$, $S_{i,1}$, and a small error subspace. we will formalize this approach in Lemma~\ref{lem:partition}.

Now, we prove the lemma which gives a decomposition of a cheating prover's state.
\begin{lemma}\label{lem:partition}
Let $(U_0,U)$ be any prover's strategy. Let $m=O(\log \secpar)$, $i\in[m]$, 
$\gamma_0 \in[0,1]$, and $T\in \mathbb{N}$ such that $\frac{\gamma_0}{T}=1/\poly(\secpar)$. Let $\gamma$ be sampled uniformly randomly from $[\frac{\gamma_0}{T},\frac{2\gamma_0}{T},\dots,\frac{T\gamma_0}{T}]$. Then, there exists an efficient quantum procedure $G_{i,\gamma}$ such that for any (possibly sub-normalized) quantum state $\ket{\psi}_{\regX,\regZ}$,  
\ifnum\submission=1
\begin{align*}
    G_{i,\gamma} \ket{0^m}_{\regC}\ket{\psi}_{\regX,\regZ}\ket{0^t}_{ph}\ket{0}_{th}\ket{0}_{in} =& z_0\ket{0^m}_{\regC}\ket{\psi_{0}}_{\regX,\regZ}\ket{0^t01}_{ph,th,in}\\ &+ z_1\ket{0^m}_{\regC}\ket{\psi_{1}}_{\regX,\regZ}\ket{0^t11}_{ph,th,in} + \ket{\psi'_{err}}
\end{align*}
\else
\begin{align*}
    G_{i,\gamma} \ket{0^m}_{\regC}\ket{\psi}_{\regX,\regZ}\ket{0^t}_{ph}\ket{0}_{th}\ket{0}_{in} = z_0\ket{0^m}_{\regC}\ket{\psi_{0}}_{\regX,\regZ}\ket{0^t01}_{ph,th,in}+ z_1\ket{0^m}_{\regC}\ket{\psi_{1}}_{\regX,\regZ}\ket{0^t11}_{ph,th,in} + \ket{\psi'_{err}}
\end{align*}
\fi
where $t$ is the number of qubits in the register $ph$, $z_0,z_1\in \mathbb{C}$ such that $|z_0|=|z_1|=1$, and 
$z_0$, $z_1$, $\ket{\psi_0}_{\regX,\regZ}$, $\ket{\psi_1}_{\regX,\regZ}$, and $\ket{\psi_{err}}_{\regX,\regZ}$ may depend on $\gamma$.

Furthermore, the following properties are satisfied.
\begin{enumerate}
    \item If we define $\ket{\psi_{err}}_{\regX,\regZ}\defeq \ket{\psi}_{\regX,\regZ} - \ket{\psi_{0}}_{\regX,\regZ}- \ket{\psi_{1}}_{\regX,\regZ}$, then we have  $E_{\gamma}[\|\ket{\psi_{err}}_{\regX,\regZ}\|^2]\leq \frac{6}{T}+\negl(n)$.
\item For any fixed $\gamma$, $\Pr[M_{ph,th,in}\circ \ket{\psi'_{err}} \in \{0^t01,0^t11\}] =0$.   
This implies that if we apply the measurement $M_{ph,th,in}$ on $\frac{G_{i,\gamma} \ket{0^m}_{\regC}\ket{\psi}_{\regX,\regZ}\ket{0^t}_{ph}\ket{0}_{th}\ket{0}_{in}}{\|\ket{\psi}_{\regX,\regZ}\|}$, then the outcome is $0^tb1$ with probability $\|\ket{\psi_b}_{\regX,\regZ}\|^2$ and the resulting state in the register $(\regX,\regZ)$  is $\frac{\ket{\psi_b}_{\regX,\regZ}}{\|\ket{\psi_b}_{\regX,\regZ}\|}$ ignoring a global phase factor.
    \item For any fixed $\gamma$, $E_{b\in \{0,1\}} [\|\ket{\psi_b}_{\regX,\regZ}\|^2]\leq \frac{1}{2}\|\ket{\psi}_{\regX,\regZ}\|^2$. 
        \item 
For any fixed $\gamma$ and $c\in \bit^m$ such that $c_i=0$, we have 
\begin{align*}
\Pr\left[M_{\regX_i}\circ U\frac{\ket{c}_{\regC}\ket{\psi_0}_{\regX,\regZ}}{\|\ket{\psi_0}_{\regX,\regZ}\|}\in \Acc_{k_i,y_i}\right]\leq 2^{m-1}\gamma+\negl(\secpar).
\end{align*}
    \item 
For any fixed $\gamma$, there exists an efficient quantum algorithm $\ext_i$ such that 
\begin{align*}  
  \Pr\left[\ext_i\left(\frac{\ket{0^m}_{\regC}\ket{\psi_1}_{\regX,\regZ}}{\|\ket{\psi_1}_{\regX,\regZ}\|}\right)\in \Acc_{k_i,y_i}\right]=1-\negl(\secpar).
  \end{align*}   
\end{enumerate}
\end{lemma}

\begin{proof}[Proof of Lemma~\ref{lem:partition}]
Procedure~\ref{fig:process_G} defines an efficient process $G_{i,\gamma}$, which decomposes $\ket{\psi}_{\regX,\regZ}$ into $\ket{\psi_0}_{\regX,\regZ}$ , $\ket{\psi_1}_{\regX,\regZ}$ , and $\ket{\psi_{err}}_{\regX,\regZ}$  described in Lemma~\ref{lem:partition}.    
Here, $G_{i,\gamma} := U_{in}U^{\dag}_{est}U_{th}U_{est}$ operates on register $\regC$, $\regX$, $\regZ$, and additional registers $ph$, $th$, and $in$, and we let $\delta:=\frac{\gamma_0}{3T}$.

\floatname{algorithm}{Procedure}
\begin{algorithm}[h]
    \begin{mdframed}[style=figstyle,innerleftmargin=10pt,innerrightmargin=10pt]
    \begin{enumerate}
    \item Do quantum phase estimation $U_{est}$ on $Q=(2\Pi_{in}-I)(2\Pi_{i,out}-I)$ with input state $\ket{0^m}_{\regC}\ket{\psi}_{\regX,\regZ}$ and $\tau$-bit precision and failure probability $2^{-n}$ where the parameter $\tau$ will be specified later, i.e.,  
    \begin{align*}
        U_{est}\ket{u}_{\regC,\regX,\regZ}\ket{0^t}_{ph} \rightarrow \sum_{\theta\in(-\pi,\pi]} \alpha_{\theta} \ket{u}_{\regC,\regX,\regZ}\ket{\theta}_{ph}.
    \end{align*}
 such that $\sum_{\theta\notin \bar{\theta}\pm 2^{-\tau}}|\alpha_{\theta}|^2\leq 2^{-\secpar}$ for any eigenvector $\ket{u}_{\regC,\regX,\regZ}$ of $Q$ with eigenvalue $e^{i\bar{\theta}}$.
    \item Apply $U_{th}:\ket{u}_{\regC,\regX,\regZ}\ket{\theta}_{ph}\ket{0}_{th} \xrightarrow{U_{th}} \ket{u}_{\regC,\regX,\regZ}\ket{\theta}_{ph}\ket{b}_{th} $, 
    where $b=1$ if $\cos^2 (\theta/2)\geq \gamma-\delta$.
    \item Apply $U_{est}^\dag$. 
   \item Apply $U_{in}: \ket{c}_{\regC}\ket{0}_{in} \xrightarrow{U_{in}}  \ket{c}_{\regC}\ket{b'}_{in}$,
    where $b'=1$ if $c= 0^m$. 
\end{enumerate}
    \caption{$G_{i,\gamma}$}
    \label{fig:process_G}
    \end{mdframed}
\end{algorithm}

In the procedure, we choose $\tau$ so that for any $\theta$ and $\theta'$ such that $|\theta'-\theta|\leq 2^{-\tau}$, we have $|\cos^2(\theta'/2)-\cos^2(\theta/2)|\leq \delta/2$.
Since we can upper and lower bound $\cos^2(\theta'/2)-\cos^2(\theta/2)$ by polynomials in $\theta'-\theta$ by considering the Taylor series, we can set $\tau=O(-\log(\delta))$ for satisfying this property.
Since phase estimation with $\tau$-bit precision and failure probability $2^{-\secpar}$ can be done in time $\poly(\secpar,2^{\tau})$ \cite{NWZ09} and $\delta=\frac{\gamma_0}{3T}=1/\poly(\secpar)$ by the assumption, the procedure runs in time $\poly(\secpar)$.

For each $j\in [\ell]$, we define $p_j:= \cos^2(\theta_j/2)=\bra{\alpha_j} \Pi_{i,out} \ket{\alpha_j}$.
We define the following projections on $\hil_{\regX,\regZ}$:
\begin{align*}
    &\Pi_{in, \leq \gamma-2\delta} := \sum_{j: p_j\leq \gamma-2\delta}\opro{\hat{\alpha}_j}{\hat{\alpha}_j}_{\regX,\regZ}+ \sum_{j} \opro{\hat{\alpha}_j^{(10)}}{\hat{\alpha}_j^{(10)}}_{\regX,\regZ},\\
     &\Pi_{in, \geq \gamma} := \sum_{j: p_j\geq  \gamma}\opro{\hat{\alpha}_j}{\hat{\alpha}_j}_{\regX,\regZ}+ \sum_{j} \opro{\hat{\alpha}_j^{(11)}}{\hat{\alpha}_j^{(11)}}_{\regX,\regZ},\\
    &\Pi_{in, mid} := \sum_{j: p_j \in (\gamma-2\delta,\gamma)}\opro{\hat{\alpha}_j}{\hat{\alpha}_j}_{\regX,\regZ}.
\end{align*}
 
We let $\ket{\psi_{\leq \gamma-2\delta}}_{\regX,\regZ}:=\Pi_{in, \leq \gamma-2\delta}\ket{\psi}_{\regX,\regZ}$,  $\ket{\psi_{\geq \gamma}}_{\regX,\regZ}:=\Pi_{in, \geq \gamma}\ket{\psi}_{\regX,\regZ}$, and $\ket{\psi_{mid}}_{\regX,\regZ}:=\Pi_{in, mid}\ket{\psi}_{\regX,\regZ}$.
Then we have 
\begin{align}
\ket{\psi}_{\regX,\regZ}=\ket{\psi_{\leq \gamma-2\delta}}_{\regX,\regZ}+\ket{\psi_{\geq\gamma}}_{\regX,\regZ}+\ket{\psi_{mid}}_{\regX,\regZ}. \label{eq:psi}
\end{align}
Roughly speaking, $\ket{\psi_{\leq \gamma-2\delta}}_{\regX,\regZ}$, $\ket{\psi_{\geq\gamma}}_{\regX,\regZ}$, $\ket{\psi_{mid}}_{\regX,\regZ}$ will correspond to $\ket{\psi_0}$, $\ket{\psi_1}$, and $\ket{\psi_{err}}$ with some error terms as explained in the following.

It is easy to see that $E_{\gamma}[\|\ket{\psi_{mid}}\|^2]\leq \frac{1}{T}$ since $\Pi_{in, mid}$ with different choice of $\gamma$ are disjoint. 
In the following, we analyze how the first two terms of Eq. \ref{eq:psi} develops by $G_{i,\gamma}$. 

$\ket{\psi_{\leq\gamma-2\delta}}_{\regX,\regZ}$ is a superposition of states $\{\ket{\hat{\alpha}_j}_{\regX,\regZ}\}_{j:p_j\leq \gamma-2\delta}$ and $\{\ket{\hat{\alpha}_j^{11}}_{\regX,\regZ}\}_{j}$.
By Lemma~\ref{lem:decomposition}, $\ket{\alpha_j}_{\regC,\regX,\regZ}=\ket{0^m}_{\regC}\ot \ket{\hat{\alpha}_j}_{\regX,\regZ}$ can be written as $\ket{\alpha_j}_{\regC,\regX,\regZ}=\frac{1}{\sqrt{2}}(\ket{\phi_j^+}_{\regC,\regX,\regZ}+\ket{\phi_j^-}_{\regC,\regX,\regZ})$ where $\ket{\phi_j^\pm}_{\regC,\regX,\regZ}$ is an eigenvector of $Q$ with eigenvalue $e^{\pm i\theta_j}$ where $\theta_j=2\arccos(\sqrt{p_j})\geq 2\arccos (\sqrt{\gamma-2\delta})$. 
Moreover,  $\ket{\alpha_j^{(10)}}_{\regC,\regX,\regZ}=\ket{0^m}_{\regC}\ot \ket{\hat{\alpha}_j^{(10)}}_{\regX,\regZ}$ is an eigenvector of $Q$ with eigenvalue $-1=e^{i\pi}$. 
Here, we remark that $\pi\geq 2\arccos x$ for any $0\leq x \leq 1$.
Thus, after applying $U_{est}$ to  $\ket{\psi_{\leq\gamma-2\delta}}_{\regX,\regZ}$, $(1-2^{-\secpar})$-fraction of the state contains $\theta$ in the register $ph$ such that $|\theta|\geq 2\arccos(\sqrt{\gamma-2\delta}) -2^{-\tau}$. which implies $\cos^2(\theta/2)\leq \gamma-\frac{3}{2}\delta<\gamma-\delta$ by our choice of $\tau$.
For this fraction of the state, $U_{th}$ does nothing.
 Thus, we have 
 \begin{align*}
 \TD(U_{est}\ket{0^m}_{\regC}\ket{\psi_{\leq\gamma-2\delta}}_{\regX,\regZ}\ket{0^t00}_{ph,th,in},U_{th}U_{est}\ket{0^m}_{\regC}\ket{\psi_{\leq\gamma-2\delta}}_{\regX,\regZ}\ket{0^t00}_{ph,th,in})\leq 2^{-n}
 \end{align*}
  and thus
  \begin{align*}
  \TD(\ket{0^m}_{\regC}\ket{\psi_{\leq\gamma-2\delta}}_{\regX,\regZ}\ket{0^t00}_{ph,th,in},U_{est}^\dag U_{th}U_{est}\ket{0^m}_{\regC}\ket{\psi_{\leq\gamma-2\delta}}_{\regX,\regZ}\ket{0^t00}_{ph,th,in})\leq 2^{-n}
  \end{align*}
   where $\TD$ denotes the trace distance. 

Therefore we can write
\begin{align}
 G_{i,\gamma}\ket{0^m}_{\regC}\ket{\psi_{\leq\gamma-2\delta}}_{\regX,\regZ}\ket{0^t00}_{ph,th,in}=z_{\leq \gamma-2\delta}\ket{0^m}_{\regC}\ket{\psi_{\leq\gamma-2\delta}}_{\regX,\regZ}\ket{0^t01}_{ph,th,in}+\ket{\psi'_{err,\leq \gamma-2\delta}} \label{eq:leq}
\end{align}
by using $z_{\leq \gamma-2\delta}$ such that $|z_{\leq \gamma-2\delta}|^2\geq 1-2^{-\secpar}$ and a state  $\ket{\psi'_{err,\leq \gamma-2\delta}}$ that is orthogonal to $\ket{0^m}_{\regC}\ket{\psi_{\leq\gamma-2\delta}}_{\regX,\regZ}\ket{0^t01}_{ph,th,in}$ such that $\|\ket{\psi'_{err,\leq \gamma-2\delta}}\|^2 \leq 2^{-n}$.

By a similar analysis, we can write 
\begin{align}
 G_{i,\gamma}\ket{0^m}_{\regC}\ket{\psi_{\geq\gamma}}_{\regX,\regZ}\ket{0^t00}_{ph,th,in}=z_{\geq \gamma}\ket{0^m}_{\regC}\ket{\psi_{\geq\gamma}}_{\regX,\regZ}\ket{0^t11}_{ph,th,in}+\ket{\psi'_{err,\geq \gamma}} \label{eq:geq}
\end{align}
by using $z_{\geq \gamma}$ such that $|z_{\geq \gamma}|^2\geq 1-2^{-\secpar}$ and a state  $\ket{\psi'_{err,\geq \gamma}}$ that is orthogonal to $\ket{0^m}_{\regC}\ket{\psi_{\geq\gamma}}_{\regX,\regZ}\ket{0^t01}_{ph,th,in}$ such that $\|\ket{\psi'_{err,\geq \gamma}}\|^2 \leq 2^{-n}$.

Combining Eq. \ref{eq:psi}, \ref{eq:leq}, and \ref{eq:geq}, we have
\begin{align}
\begin{split}
& G_{i,\gamma}\ket{0^m}_{\regC}\ket{\psi}_{\regX,\regZ}\ket{0^t00}_{ph,th,in} \\
&=\ket{0^m}_{\regC}(z_{\geq \gamma}\ket{\psi_{\geq\gamma}}_{\regX,\regZ}+\ket{\eta_{mid,0}}_{\regX,\regZ}+\ket{\eta_{other,0}}_{\regX,\regZ})\ket{0^t01}_{ph,th,in} \\
&+\ket{0^m}_{\regC}(z_{\leq \gamma-2\delta}\ket{\psi_{\leq\gamma-2\delta}}_{\regX,\regZ}+\ket{\eta_{mid,1}}_{\regX,\regZ}+\ket{\eta_{other,1}}_{\regX,\regZ})\ket{0^t11}_{ph,th,in} \\
&+ \ket{\psi'_{err}}. 
\end{split} \label{eq:resulting_state}
\end{align}

where $\ket{\eta_{mid,0}}_{\regX,\regZ}$, $\ket{\eta_{other,0}}_{\regX,\regZ}$, $\ket{\eta_{mid,1}}_{\regX,\regZ}$, and $\ket{\eta_{other,1}}_{\regX,\regZ}$ are defined so that
\begin{align}
& I_{\regC,\regX,\regZ} \ot \opro{0^t01}{0^t01}_{ph,th,in}G_{i,\gamma}\ket{0^m}_{\regC}\ket{\psi_{mid}}_{\regX,\regZ}\ket{0^t00}_{ph,th,in}=\ket{0^m}_{\regC}{\ket{\eta_{mid,0}}_{\regX,\regZ}}\ket{0^t01}_{ph,th,in}, \label{eq:etamid} \\
&I_{\regC,\regX,\regZ} \ot \opro{0^t11}{0^t11}_{ph,th,in}G_{i,\gamma}\ket{0^m}_{\regC}\ket{\psi_{mid}}_{\regX,\regZ}\ket{0^t00}_{ph,th,in}=\ket{0^m}_{\regC}{\ket{\eta_{mid,1}}_{\regX,\regZ}}\ket{0^t11}_{ph,th,in}, \label{eq:etamidprime} \\
&I_{\regC,\regX,\regZ} \ot \opro{0^t01}{0^t01}_{ph,th,in}(\ket{\psi'_{err,\leq \gamma-2\delta}}+\ket{\psi'_{err,\geq \gamma}})=\ket{0^m}_{\regC}{\ket{\eta_{other,0}}_{\regX,\regZ}}\ket{0^t01}_{ph,th,in}, \label{eq:etaother} \\
&I_{\regC,\regX,\regZ} \ot \opro{0^t11}{0^t11}_{ph,th,in}(\ket{\psi'_{err,\leq \gamma-2\delta}}+\ket{\psi'_{err,\geq \gamma}})=\ket{0^m}_{\regC}{\ket{\eta_{other,1}}_{\regX,\regZ}}\ket{0^t11}_{ph,th,in}. \label{eq:etaotherprime} 
\end{align}
and $\ket{\psi'_{err}}$  is defined by 
\ifnum\submission=1
\begin{align}
\begin{split}
\ket{\psi'_{err}}:=
\sum_{s\notin \{0^t01,0^t11\}} I_{\regC,\regX,\regZ} \ot \opro{s}{s}_{ph,th,in}
(&G_{i,\gamma}\ket{0^m}_{\regC}\ket{\psi_{mid}}_{\regX,\regZ}\ket{0^t00}_{ph,th,in}\\
&+\ket{\psi'_{err,\leq \gamma-2\delta}}+\ket{\psi'_{err,\geq \gamma}}). 
\end{split}
\label{eq:psierrprime}
\end{align}
\else
\begin{align}
\ket{\psi'_{err}}&:=
\sum_{s\notin \{0^t01,0^t11\}} I_{\regC,\regX,\regZ} \ot \opro{s}{s}_{ph,th,in}(G_{i,\gamma}\ket{0^m}_{\regC}\ket{\psi_{mid}}_{\regX,\regZ}\ket{0^t00}_{ph,th,in}+\ket{\psi'_{err,\leq \gamma-2\delta}}+\ket{\psi'_{err,\geq \gamma}}). \label{eq:psierrprime}
\end{align}
\fi

We remark that $\ket{\eta_{mid,0}}_{\regX,\regZ}$, $\ket{\eta_{other,0}}_{\regX,\regZ}$, $\ket{\eta_{mid,1}}_{\regX,\regZ}$, and $\ket{\eta_{other,1}}_{\regX,\regZ}$ are well-defined since after applying $G_{i,\gamma}$, the value in the register $in$ is $1$ if and only if the value in the register $\regC$ is $0^m$.

We let $z_0:=\frac{z_{\leq \gamma-2\delta}}{|z_{\leq \gamma-2\delta}|}$, $z_1:=\frac{z_{\geq \gamma}}{|z_{\geq \gamma}|}$, and 
\begin{align}
&\ket{\psi_0}_{\regX,\regZ}:=|z_{\leq \gamma-2\delta}|\ket{\psi_{\leq\gamma-2\delta}}_{\regX,\regZ}+ \overline{z}_0(\ket{\eta_{mid,0}}_{\regX,\regZ}+\ket{\eta_{other,0}}_{\regX,\regZ}), \label{eq:psizero} \\
&\ket{\psi_1}_{\regX,\regZ}:=|z_{\geq \gamma}|\ket{\psi_{\geq\gamma}}_{\regX,\regZ}+ \overline{z}_1(\ket{\eta_{mid,1}}_{\regX,\regZ}+\ket{\eta_{other,1}}_{\regX,\regZ}), \label{eq:psione}
\end{align}
where $\overline{z}_0$ and $\overline{z}_1$ denotes complex conjugates of $z_0$ and $z_1$. By Eq~\ref{eq:resulting_state}, \ref{eq:psizero}, and \ref{eq:psione}, we have
\ifnum\submission=1
\begin{align*}
    G_{i,\gamma} \ket{0^m}_{\regC}\ket{\psi}_{\regX,\regZ}\ket{0^t}_{ph}\ket{0}_{th}\ket{0}_{in} = &z_0\ket{0^m}_{\regC}\ket{\psi_{0}}_{\regX,\regZ}\ket{0^t01}_{ph,th,in}\\ &+ z_1\ket{0^m}_{\regC}\ket{\psi_{1}}_{\regX,\regZ}\ket{0^t11}_{ph,th,in} + \ket{\psi'_{err}}.
\end{align*}
\else
\begin{align*}
    G_{i,\gamma} \ket{0^m}_{\regC}\ket{\psi}_{\regX,\regZ}\ket{0^t}_{ph}\ket{0}_{th}\ket{0}_{in} = z_0\ket{0^m}_{\regC}\ket{\psi_{0}}_{\regX,\regZ}\ket{0^t01}_{ph,th,in}+ z_1\ket{0^m}_{\regC}\ket{\psi_{1}}_{\regX,\regZ}\ket{0^t11}_{ph,th,in} + \ket{\psi'_{err}}.
\end{align*}
\fi

Now, we are ready to prove the five claims in Lemma~\ref{lem:partition}.

\paragraph{Proof of the first claim.}
By Eq.~\ref{eq:psi}, \ref{eq:psizero}, and \ref{eq:psione}, we have
\begin{align*}
\ket{\psi_{err}}_{\regX,\regZ}&=\ket{\psi}_{\regX,\regZ}-\ket{\psi_0}_{\regX,\regZ}-\ket{\psi_1}_{\regX,\regZ}\\
&=(1-|z_{\leq \gamma-2\delta}|)\ket{\psi_{\leq\gamma-2\delta}}_{\regX,\regZ}+(1-|z_{\geq \gamma}|)\ket{\psi_{\geq\gamma}}_{\regX,\regZ}+\ket{\psi_{mid}}_{\regX,\regZ}\\
&-\overline{z}_0(\ket{\eta_{mid,0}}_{\regX,\regZ}+\ket{\eta_{other,0}}_{\regX,\regZ})
-\overline{z}_1(\ket{\eta_{mid,1}}_{\regX,\regZ}+\ket{\eta_{other,1}}_{\regX,\regZ}),
\end{align*}

Since $|z_{\leq \gamma-2\delta}|$ and $|z_{\geq \gamma}|$ are $1-\negl(\secpar)$, the norms of the first two terms are negligible.
By Eq.~\ref{eq:etaother} and \ref{eq:etaotherprime}, we have $\|\ket{\eta_{other,0}}_{\regX,\regZ}\|^2 + \|\ket{\eta_{other,1}}_{\regX,\regZ}\|^2\leq \|\ket{\psi'_{err,\leq \gamma-2\delta}}+\ket{\psi'_{err,\geq \gamma}}\|^2 \leq \negl(\secpar)$.
Therefore we have
\begin{align*}
\|\ket{\psi_{err}}_{\regX,\regZ}\|^2 &\leq \|\ket{\psi_{mid}}_{\regX,\regZ}-\overline{z}_0\ket{\eta_{mid,0}}_{\regX,\regZ}-\overline{z}_1\ket{\eta_{mid,1}}_{\regX,\regZ}\|^2+\negl(\secpar)\\
&\leq 3(\|\ket{\psi_{mid}}_{\regX,\regZ}\|^2+\|\ket{\eta_{mid,0}}_{\regX,\regZ}\|^2 + \|\ket{\eta_{mid,1}}_{\regX,\regZ}\|^2)+\negl(\secpar)
\end{align*}
where the latter inequality follows from the Cauchy-Schwarz inequality.
As already noted, we have $E_{\gamma}[\|\ket{\psi_{mid}}_{\regX,\regZ}\|^2]\leq \frac{1}{T}$.
By Eq.~\ref{eq:etamid} and \ref{eq:etamidprime}, we have $E_{\gamma}[\|\ket{\eta_{mid,0}}_{\regX,\regZ}\|^2 \allowbreak + \|\ket{\eta_{mid,1}}_{\regX,\regZ}\|^2]\leq E_{\gamma}[\|\ket{\psi_{mid}}_{\regX,\regZ}\|^2]\leq \frac{1}{T}$.
Therefore, we have $E_{\gamma}[\|\ket{\psi_{err}}_{\regX,\regZ}\|^2]\leq \frac{6}{T}+\negl(\secpar)$ and the first claim is proven.

\paragraph{Proof of the second claim.}
By Eq~\ref{eq:psierrprime}, we can see that 
\begin{align*}
I_{\regC,\regX,\regZ} \ot (\opro{001}{001}_{ph,th,in}+ \opro{011}{011}_{ph,th,in})\ket{\psi'_{err}}=0.
\end{align*}
This immediately implies the second claim.
\paragraph{Proof of the third claim.}
By the second claim,  $\ket{0^m}_{\regC}\ket{\psi_{0}}_{\regX,\regZ}\ket{0^t01}_{ph,th,in}$,  $\ket{0^m}_{\regC}\ket{\psi_{1}}_{\regX,\regZ}\ket{0^t11}_{ph,th,in}$, and $\ket{\psi'_{err}}$ are orthogonal with one another.
Therefore we have 
\begin{align*}
&\|G_{i,\gamma}\ket{0^m}_{\regC}\ket{\psi}_{\regX,\regZ}\ket{0^t00}_{ph,th,in}\|^2\\
=&\|z_0\ket{0^m}_{\regC}\ket{\psi_{0}}_{\regX,\regZ}\ket{0^t01}_{ph,th,in}\|^2+\|z_1\ket{0^m}_{\regC}\ket{\psi_{1}}_{\regX,\regZ}\ket{0^t11}_{ph,th,in}\|^2+\|\ket{\psi'_{err}}\|^2.
\end{align*}
Since we have 
$\|G_{i,\gamma}\ket{0^m}_{\regC}\ket{\psi}_{\regX,\regZ}\ket{0^t00}_{ph,th,in}\|^2=\|\ket{\psi}_{\regX,\regZ}\|^2$ and $\|z_b\ket{0^m}_{\regC}\ket{\psi_{b}}_{\regX,\regZ}\ket{0^tb1}_{ph,th,in}\|^2 \allowbreak =\|\ket{\psi_{b}}_{\regX,\regZ}\|^2$, the above implies $\|\ket{\psi_{0}}_{\regX,\regZ}\|^2+\|\ket{\psi_{1}}_{\regX,\regZ}\|^2\leq \|\ket{\psi}_{\regX,\regZ}\|^2$, which implies the third claim.
\paragraph{Proof of the forth claim.}
Roughly speaking, we first show that $\ket{0^m}_{\regC}\ket{\psi_0}_{\regX,\regZ}$ 
 can be written as a superposition of states $\{\ket{\alpha_j}\}_{j: p_j\leq \gamma}$ and $\{\ket{\alpha_j^{(10)}}\}_{j}$ except for a term with a negligible norm. 
Since each of the above states has eigenvalues at most $\gamma$ w.r.t. $\Pi_{i,out}$, we can show that $\|\Pi_{i,out}\ket{0^m}_{\regC}\ket{\psi_0}_{\regX,\regZ}\|^2$ is at most $\gamma+\negl(\secpar)$. 
This means that $\ket{\psi_0}_{\regX,\regZ}$ leads to acceptance in the test round on the $i$-th coordinate with probability at most $\gamma+\negl(\secpar)$ if $c_{-i}\in \bit^{m-1}$ is chosen randomly.
Since the number of possible choices of $c$ is $2^{m-1}$, the probability is at most $2^{m-1}\gamma+\negl(\secpar)$ for any fixed $c$, which implies the forth claim.  The detail follows.
 
 We analyze each term of Eq. \ref{eq:psizero} separately.
 First, since we have $\ket{\psi_{\leq \gamma-2\delta}}_{\regX,\regZ}=\Pi_{in, \leq \gamma-2\delta}\ket{\psi}_{\regX,\regZ}$, $\ket{\psi_{\leq \gamma-2\delta}}_{\regX,\regZ}$ is a superposition of states $\{{\ket{\hat{\alpha}_j}}\}_{j: p_j\leq \gamma-2\delta}$ and $\{\ket{\hat{\alpha}_j^{(10)}}\}_{j}$ by the definition of $\Pi_{in, \leq \gamma-2\delta}$.
 Therefore, $\ket{0^m}_{\regC}\ket{\psi_{\leq \gamma-2\delta}}_{\regX,\regZ}$ is a superposition of states $\{{\ket{\alpha_j}}\}_{j: p_j\leq \gamma-2\delta}$ and $\{\ket{\alpha_j^{(10)}}\}_{j}$

 Second, we analyze $\ket{\eta_{mid,0}}$.
 By the definition of $\ket{\psi_{mid}}_{\regX,\regZ}$, the state $\ket{0^m}_{\regC}\ket{\psi_{mid}}_{\regX,\regZ}$ is in the subspace $S_{mid}$, which is the subspace spanned by $\{S_j\}_{j:p_j\in (\gamma-2\delta,\gamma)}$.
We define $\ket{\psi''_{mid,s}}_{\regC,\regX,\regZ}$ so that 
\begin{align*}
G_{i,\gamma}\ket{0^m}_{\regC}\ket{\psi_{mid}}\ket{0^t00}_{ph,th,in}=\sum_{s\in \bit^{t+2}} \ket{\psi''_{mid,s}}_{\regC,\regX,\regZ}\ket{s}_{ph,th,in}.
\end{align*}
Since each subspace $S_j$ is invariant under the projections $\Pi_{in}$ and $\Pi_{i,out}$, each $\ket{\psi''_{mid,s}}_{\regC,\regX,\regZ}$ is also in the subspace $S_{mid}$. 
In particular, $\ket{0^m}_{\regC}\ket{\eta_{mid,0}}_{\regX,\regZ}=\ket{\psi''_{mid,0^t01}}_{\regC,\regX,\regZ}$ is in the subspace $S_{mid}$.
That is, $\ket{0^m}_{\regC}\ket{\eta_{mid,0}}_{\regX,\regZ}$ is a superposition of $\{\ket{\alpha_j}\}_{j: p_j\in (\gamma-2\delta, \gamma)}$.

Third, we can see that $\|\ket{\eta_{other,0}}_{\regX,\regZ}\|=\negl(\secpar)$ from Eq. \ref{eq:etaother} and that $\|\ket{\psi'_{err,\leq \gamma-2\delta}}+\ket{\psi'_{err,\geq \gamma}}\|=\negl(\secpar)$.

Combining the above together with Eq. \ref{eq:psizero}, we can write 
\begin{align}
\ket{0^m}_{\regC}\ket{\psi_0}_{\regX,\regZ}=\sum_{j:p_j<\gamma} d_j \ket{\alpha_j}_{\regC,\regX,\regZ}+\sum_{j} d_j^{(10)} \ket{\alpha_j^{(10)}}_{\regC,\regX,\regZ}+\ket{\psi''_{err,0}}_{\regC,\regX,\regZ}  \label{eq:psizero_decompose} 
\end{align}
for some $d_j,d_j^{(10)}\in \mathbb{C}$ and a state $\ket{\psi''_{err,0}}_{\regC,\regX,\regZ}:=\overline{z}_0 \ket{0^m}_{\regC}\ket{\eta_{other,0}}_{\regX,\regZ}$ whose norm is $\negl(\secpar)$.
Here, we remark that the first term comes from 
both $|z_{\leq \gamma-2\delta}|\ket{0^m}_{\regC}\ket{\psi_{\leq\gamma-2\delta}}_{\regX,\regZ}$ and $\overline{z}_0\ket{0^m}_{\regC}\ket{\eta_{mid,0}}$, and the second term comes from $|z_{\leq \gamma-2\delta}|\ket{0^m}_{\regC}\ket{\psi_{\leq\gamma-2\delta}}_{\regX,\regZ}$.

By the definition of $\Pi_{i,out}$, we have 
\begin{align}
 \Pr_{c_{-i}}\left[M_{\regX_i}\circ U\frac{\ket{c_1...c_{i-1}0c_{i+1}...c_m}_{\regC}\ket{\psi_0}_{\regX,\regZ}}{\|\ket{\psi_0}_{\regX,\regZ}\|}\in \Acc_{k_i,y_i}\right]=\frac{\bra{0^m}_{\regC}\bra{\psi_0}_{\regX,\regZ} \Pi_{i,out} \ket{0^m}_{\regC}\ket{\psi_0}_{\regX,\regZ}}{\|\ket{\psi_0}_{\regX,\regZ}\|^2} \label{eq:accept_probability}
\end{align}
where $c_{-i}$ denotes $c_1...c_{i-1}c_{i+1}...c_{m}$.

By Lemma~\ref{lem:decomposition}, we can see that $\bra{\alpha_j}\Pi_{i,out} \ket{\alpha_{j'}}=0$ for all $j\neq j'$ and $\Pi_{i,out} \ket{\alpha_{j^{(10)}}}=0$ for all $j$.
By substituting Eq. \ref{eq:psizero_decompose} for Eq. \ref{eq:accept_probability}, we have
\begin{align*}
&~~~\Pr_{c_{-i}}\left[M_{\regX_i}\circ U\frac{\ket{c_1...c_{i-1}0c_{i+1}...c_m}_{\regC}\ket{\psi_0}_{\regX,\regZ}}{\|\ket{\psi_0}_{\regX,\regZ}\|}\in \Acc_{k_i,y_i}\right]\\
&=\frac{1}{\|\ket{\psi_0}_{\regX,\regZ}\|^2}\left(\sum_{j:p_j<\gamma} |d_j|^2 \bra{\alpha_j} \Pi_{i,out} \ket{\alpha_j}+ \sum_{j:p_j<\gamma} (\overline{d}_j \bra{\alpha_j}\Pi_{i,out} \ket{\psi''_{err,0}} + d_j \bra{\psi''_{err,0}}\Pi_{i,out} \ket{\alpha_j})\right)\\
&\leq \gamma+\negl(\secpar)
\end{align*}
where the last inequality follows from $\sum_{j:p_j<\gamma}|d_j|^2\leq \|\ket{\psi_0}_{\regX,\regZ}\|^2$ and $\|\ket{\psi''_{err,0}}_{\regC,\regX,\regZ}\|=\negl(\secpar)$. 
This immediately implies the forth claim considering that the number of possible $c_{-i}$ is $2^{m-1}$ and $m=O(\log \secpar)$.
\paragraph{Proof of the fifth claim.}
By a similar argument to the one in the proof of the forth claim, we can write  
\begin{align}
\ket{0^m}_{\regC}\ket{\psi_1}_{\regX,\regZ}=\sum_{j:p_j>\gamma-2\delta} d_j \ket{\alpha_j}_{\regC,\regX,\regZ}+\sum_{j} `d_j^{(11)} \ket{\alpha_j^{(11)}}_{\regC,\regX,\regZ}+\ket{\psi''_{err,1}}_{\regC,\regX,\regZ}  \label{eq:psione_decompose} 
\end{align}
where $\ket{\psi''_{err,1}}_{\regC,\regX,\regZ}$ is a state such that $\|\ket{\psi''_{err,1}}_{\regC,\regX,\regZ}\|=\negl(\secpar)$.

The algorithm $\ext_{i}$ is described below:

\begin{description}
\item[$\ext_{i}\left(\frac{\ket{0^m}_{\regC}\ket{\psi_1}_{\regX,\regZ}}{\|\ket{\psi_1}_{\regX,\regZ}\|}\right)$:]
Given $\frac{\ket{0^m}_{\regC}\ket{\psi_1}_{\regX,\regZ}}{\|\ket{\psi_1}_{\regX,\regZ}\|}$ as input, $\ext_{i}$ works as follows:
\begin{itemize}
\item Repeat the following procedure $N=\poly(\secpar)$ times where $N$ is specified later:
\begin{enumerate}
\item Perform a measurement $\{\Pi_{i,out},I_{\regC,\regX,\regZ}-\Pi_{i,out}\}$. If the outcome is $0$, i,e, $\Pi_{i,out}$ is applied, then measure the register $\regX_i$ in computational basis to obtain $a_i$, outputs $a_i$, and immediately halts.
\item  Perform a measurement $\{\Pi_{in},I_{\regC,\regX,\regZ}-\Pi_{in}\}$.
\end{enumerate}
\item If it does not halts within $N$ trials in the previous step, output $\bot$.
\end{itemize}
\end{description}

By the definition of $\Pi_{i,out}$, it is clear that $\ext_{i}$ succeeds, (i.e., returns $a_i\in \Acc_{k_i,y_i}$) if it does not output $\bot$.
Since the algorithm $\ext_{i}$ just alternately applies measurements $\{\Pi_{i,out},I_{\regC,\regX,\regZ}-\Pi_{i,out}\}$ and $\{\Pi_{in},I_{\regC,\regX,\regZ}-\Pi_{in}\}$ and each subspaces $S_j$ and $T_j^{(11)}$ are invariant under $\Pi_{in}$ and $\Pi_{i,out}$, we can analyze the success probability of the algorithm separately on each subspace.
Therefore, it suffices to show that $\ext_{i}$ succeeds with probability $1-\negl(\secpar)$ on any input $\ket{\alpha_j}_{\regX,\regZ}$ such that  $p_j> \gamma-2\delta$ or $\ket{\alpha_j^{(11)}}$ for any $j$.
First, it is easy to see that on input $\ket{\alpha_j^{(11)}}$, $\ext_{i}$ returns $a_i\in \Acc_{k_i,y_i}$ at the first trial with probability $1$ since we have $\bra{\alpha_{j}^{(11)}} \Pi_{i,out} \ket{\alpha_{j}^{11}}=1$.
What is left is to prove that $\ext_{i}$ succeeds with probability $1-\negl(\secpar)$ on any input $\ket{\alpha_j}_{\regX,\regZ}$ such that  $p_j> \gamma-2\delta$. 

By Lemma~\ref{lem:decomposition}, it is easy to see that we have
\begin{align*}
&\ket{\alpha_j}_{\regX,\regZ}=\sqrt{p_j}\ket{\beta_j}_{\regX,\regZ}+\sqrt{1-p_j}\ket{\beta_j^\bot}_{\regX,\regZ},\\
&\ket{\beta_j}_{\regX,\regZ}=\sqrt{p_j}\ket{\alpha_j}_{\regX,\regZ}+\sqrt{1-p_j}\ket{\alpha_j^\bot}_{\regX,\regZ}.
\end{align*}

Let  $P_k$ and $P_k^\bot$ be the probability that $\ext_i$ succeeds within $k$ trials starting from the initial state $\ket{\alpha_j}_{\regX,\regZ}$   and $\ket{\alpha_j^\bot}_{\regX,\regZ}$, respectively.
Then by the above equations, it is easy to see that we have $P_0=P_0^\bot=0$ and
\begin{align*}
&P_{k+1}=p_j+(1-p_j)^2 P_{k}+ (1-p_j)p_j P_{k}^\bot, \\
&P_{k+1}^\bot=(1-p_j)+ p_j(1-p_j) P_{k}+ p_j^2 P_{k}^\bot.
\end{align*}

Solving this, we have 
\begin{align*}
P_N=1-(1-2p_j+2p_j^2)^{N-1}(1-p_j).
\end{align*}

Since we assume $p_j> \gamma-2\delta>\frac{\gamma_0}{3T}=1/\poly(\secpar)$, we have $1-2p_j+2p_j^2=1-1/\poly(\secpar)$.
Therefore if we take $N=\poly(\secpar)$ sufficiently large, then $P_N=1-\negl(\secpar)$.
This means that $\ext_i$ succeeds within $N$ steps with probability $1-\negl(\secpar)$ starting from the initial state $\ket{\alpha_j}_{\regX,\regZ}$.
This completes the proof of the fifth claim and the proof of Lemma~\ref{lem:partition}. 
\end{proof}

In Lemma~\ref{lem:partition}, we showed that by fixing any $i\in [m]$, we can partition any prover's state $\ket{\psi}_{\regX,\regZ}$ into $\ket{\psi_0}_{\regX,\regZ}$, $\ket{\psi_1}_{\regX,\regZ}$, and $\ket{\psi_{err}}_{\regX,\regZ}$ with certain properties. 
In the following, we sequentially apply Lemma~\ref{lem:partition} for each $i\in[m]$ to further decompose the prover's state.

\begin{lemma}\label{lem:partition_further}
Let $m$, $\gamma_0$, $T$ be as in Lemma~\ref{lem:partition}, and let $\gamma_i\sample [\frac{\gamma_0}{T},\frac{2\gamma_0}{T},\dots,\frac{T\gamma_0}{T}]$ for each $i\in [m]$.
For any $c\in \bit^m$, a state $\ket{\psi}_{\regX,\regZ}$ can be partitioned as follows.
\begin{align*}
    & \ket{\psi}_{\regX,\regZ} = \ket{\psi_{c_1}}_{\regX,\regZ} + \ket{\psi_{\bar{c}_1,c_2}}_{\regX,\regZ} + \cdots +\ket{\psi_{\bar{c}_1,\dots,\bar{c}_{m-1},c_m}}_{\regX,\regZ} + \ket{\psi_{\bar{c}_1,\dots,\bar{c}_m}}_{\regX,\regZ}+ \ket{\psi_{err}}_{\regX,\regZ}
\end{align*}
where the way of partition may depend on the choice of $\hat{\gamma}=\gamma_1...\gamma_m$.
Further, the following properties are satisfied. 
\begin{enumerate}
    \item For any fixed $\hat{\gamma}$ and any $c$, $i\in [m]$ such that $c_i=0$, we have 
    \begin{align*}
    \Pr\left[M_{\regX_i}\circ U \frac{\ket{0^m}_{\regC}\ket{\psi_{\bar{c}_1,\dots,\bar{c}_{i-1},0}}_{\regX,\regZ}}{|\ket{\psi_{\bar{c}_1,\dots,\bar{c}_{i-1},0}}_{\regX,\regZ}|}\in \Acc_{k_i,y_i}\right]\leq 2^{m-1}\gamma_0+ \negl(\secpar).
    \end{align*}
    
    \item For any fixed $\hat{\gamma}$ and any $c$, $i\in[m]$ such that $c_i=1$, there exists an efficient algorithm $\ext_i$ such that 
    \begin{align*}  
  \Pr\left[\ext_i\left(\frac{\ket{0^m}_{\regC}\ket{\psi_{\bar{c}_1,\dots,\bar{c}_{i-1},1}}_{\regX,\regZ}}{\|\ket{\psi_{\bar{c}_1,\dots,\bar{c}_{i-1},1}}_{\regX,\regZ}\|}\right)\in \Acc_{k_i,y_i}\right]=1-\negl(\secpar).
  \end{align*}   
  \item For any fixed $\hat{\gamma}$, we have $E_c[\|\ket{\psi_{\bar{c}_1,\dots,\bar{c}_m}}_{\regX,\regZ}\|^2] \leq 2^{-m}$.
\item For any fixed $c$, we have $E_{\hat{\gamma}}[\|\ket{\psi_{err}}_{\regX,\regZ}\|^2]\leq \frac{6m^2}{T}+\negl(\secpar)$.
    \item For any fixed $\hat{\gamma}$ and $c$ there exists an efficient quantum algorithm $H_{\hat{\gamma},c}$ that is given $\ket{\psi}_{\regX,\regZ}$ as input and produces  $\frac{\ket{\psi_{\bar{c}_1,\dots,\bar{c}_{i-1},c_i}}_{\regX,\regZ}}{\|\ket{\psi_{\bar{c}_1,\dots,\bar{c}_{i-1},c_i}}_{\regX,\regZ}\|}$ with probability $\|\ket{\psi_{\bar{c}_1,\dots,\bar{c}_{i-1},c_i}}_{\regX,\regZ}\|^2$ ignoring a global phase factor.
\end{enumerate}
\end{lemma}
\begin{proof}
We inductively define $\ket{\psi_{c_1}}_{\regX,\regZ}$,...,$\ket{\psi_{\bar{c}_1,...,\bar{c}_m}}_{\regX,\regZ}$ as follows.

First, we apply Lemma \ref{lem:partition} for the state $\ket{\psi}_{\regX,\regZ}$ with $\gamma=\gamma_1$ to give a decomposition
\begin{align*}
\ket{\psi}_{\regX,\regZ}=\ket{\psi_0}_{\regX,\regZ}+\ket{\psi_1}_{\regX,\regZ} + \ket{\psi_{err,1}}_{\regX,\regZ} 
\end{align*}
where $\ket{\psi_{err,1}}_{\regX,\regZ}$ corresponds to $\ket{\psi_{err}}_{\regX,\regZ}$ in Lemma \ref{lem:partition}.

For each $i=2,...,m$, we apply  Lemma \ref{lem:partition} for the state $\ket{\psi_{\bar{c}_1,...,\bar{c}_{i-1}}}_{\regX,\regZ}$ with $\gamma=\gamma_i$ to give a decomposition 
\begin{align*}
\ket{\psi_{\bar{c}_1,...,\bar{c}_{i-1}}}_{\regX,\regZ}=\ket{\psi_{\bar{c}_1,...,\bar{c}_{i-1},0}}_{\regX,\regZ}+\ket{\psi_{\bar{c}_1,...,\bar{c}_{i-1},1}}_{\regX,\regZ} + \ket{\psi_{err,i}}_{\regX,\regZ} 
\end{align*}  
where 
$\ket{\psi_{\bar{c}_1,...,\bar{c}_{i-1},0}}_{\regX,\regZ}$, $\ket{\psi_{\bar{c}_1,...,\bar{c}_{i-1},1}}_{\regX,\regZ}$, and $\ket{\psi_{err,i}}_{\regX,\regZ}$ corresponds to $\ket{\psi_{0}}_{\regX,\regZ}$, $\ket{\psi_{1}}_{\regX,\regZ}$, and $\ket{\psi_{err}}_{\regX,\regZ}$ in Lemma \ref{lem:partition}, respectively.
  
 Then it is easy to see that we have
 \begin{align*}
    & \ket{\psi}_{\regX,\regZ} = \ket{\psi_{c_1}}_{\regX,\regZ} + \ket{\psi_{\bar{c}_1,c_2}}_{\regX,\regZ} + \cdots +\ket{\psi_{\bar{c}_1,\dots,\bar{c}_{m-1},c_m}}_{\regX,\regZ} + \ket{\psi_{\bar{c}_1,\dots,\bar{c}_m}}_{\regX,\regZ}+ \ket{\psi_{err}}_{\regX,\regZ}
\end{align*}
where we define $\ket{\psi_{err}}_{\regX,\regZ}\defeq \sum_{i=1}^{m}\ket{\psi_{err,i}}_{\regX,\regZ}$. 
  
The first and second claims immediately follow from the forth and fifth claims of Lemma~\ref{lem:partition} and $\gamma_i\leq \gamma_0$ for each $i\in[m]$.  

By the third claim of  Lemma~\ref{lem:partition}, we have $E_{c_1...c_{i}}[\|\ket{\psi_{\bar{c}_1,...,\bar{c}_{i}}}_{\regX,\regZ}\|]\leq \frac{1}{2}E_{c_1...c_{i-1}}[\|\ket{\psi_{\bar{c}_1,...,\bar{c}_{i-1}}}_{\regX,\regZ}\|]$.
Ths implies the third claim.

By the first claim of  Lemma~\ref{lem:partition}, we have $E_{\gamma_i}[\|\ket{\psi_{err,i}}_{\regX,\regZ}\|^2]\leq \frac{6}{T}+\negl(\secpar)$.
The forth claim follows from this and the Cauchy-Schwarz inequality. 

Finally, for proving the fifth claim, we define the procedure $H_{\hat{\gamma},c}$ as described in Procedure~\ref{fig:process_H}
We can easily see that $H_{\hat{\gamma},c}$ satisfies the desired property by the second claim of  Lemma~\ref{lem:partition}.

\floatname{algorithm}{Procedure}
\begin{algorithm}[h]
    \begin{mdframed}[style=figstyle,innerleftmargin=10pt,innerrightmargin=10pt]
   On input $\ket{\psi}_{\regX,\regZ}$, it works as follows:
   
   For each $i=1,...,m$, it applies 
    \begin{enumerate}
    \item Prepare registers $\regC$, $(ph_1,th_1,in_1)$,..., $(ph_m,th_m,in_m)$ all of which are initialized to be $\ket{0}$.
    \item For each $i=1,...,m$, do the following: 
      \begin{enumerate}
      \item Apply $G_{i,\gamma_i}$ on the quantum state in the registers $(\regC,\regX,\regZ,ph_i,th_i,in_i)$.
      \item Measure the registers $(ph_i,th_i,in_i)$ in the computational basis.
      \item If the outcome is $0^tc_{i}1$, then it halts and returns the state in the register $(\regX,\regZ)$. If the outcome is $0^t\bar{c}_{i}1$, continue to run. Otherwise, immediately halt and abort.
      \end{enumerate}   
    \end{enumerate}

    \caption{$H_{\hat{\gamma},c}$}
    \label{fig:process_H}
    \end{mdframed}
\end{algorithm}
\end{proof}

Given Lemma~\ref{lem:partition_further}, we can start proving Theorem~\ref{thm:rep_soundness}. 

\begin{proof}[Proof of Theorem~\ref{thm:rep_soundness}]
 
First, we recall how a cheating prover characterized by $(U_0,U)$ works.
When the first message $k$ is given, it first applies 
\begin{align*}
    &U_0\ket{0}_{\regX,\regZ}\ket{0}_{\regY}\ket{k}_{\regK} \xrightarrow{\mbox{measure }\regY} \ket{\psi(k,y)}_{\regX,\regZ}\ket{k}_{\regK}.
\end{align*}
to generate the second message $y$ and $\ket{\psi(k,y)}_{\regX,\regZ}$.
Then after receiving the third message $c$, it applies $U$ on $\ket{c}_{\regC}\ket{\psi(k,y)}_{\regX,\regZ}$ and measures the register $\regX$ in the computational basis to obtain the forth message $a$.
In the following, we just write $\ket{\psi}_{\regX,\regZ}$ to mean $\ket{\psi(k,y)}_{\regX,\regZ}$ for notational simplicity.
Let $M_{i,k_i,\td_i,y_i,c_i}$ be the measurement that outputs the verification result of the value in the register $\regX_i$ w.r.t.  $k_i,\td_i,y_i,c_i$, and let $M_{k,\td,y,c}$ be the measurement that returns $\top$ if and only if $M_{i,k_i,\td_i,y_i,c_i}$ returns $\top$ for all $i\in[m]$ where $k=(k_1,...,k_m)$, $\td=(\td_1,...,\td_m)$, $y=(y_1,...,y_m)$ and $c=(c_1,...,c_m)$.
With this notation, a cheating prover's success probability can be written as 
\begin{align*}
    \Pr_{k,\td,y,c}[M_{k,\td,y,c}U\ket{c}_{\regC}\ket{\psi}_{\regX,\regZ} = \top].
\end{align*}

Let $\gamma_0$, $\hat{\gamma}$, and $T$ be as in Lemma~\ref{lem:partition_further}.
According to Lemma~\ref{lem:partition_further}, for any fixed $\hat{\gamma}$ and $c\in \bit^{m}$, we can decompose $\ket{\psi}_{\regX,\regZ}$ as 
\begin{align*}
    \ket{\psi}_{\regX,\regZ} =  \ket{\psi_{c_1}}_{\regX,\regZ}+ \ket{\psi_{\bar{c}_1,c_2}}_{\regX,\regZ} + \cdots + \ket{\psi_{\bar{c}_1,\dots, \bar{c}_{m-1},c_{m}}}_{\regX,\regZ} + \ket{\psi_{\bar{c}_1,\dots, \bar{c}_{m-1},\bar{c}_{m}}}_{\regX,\regZ}+ \ket{\psi_{err}}_{\regX,\regZ}.
\end{align*}

To prove the theorem, we show the following two inequalities.
First,  for any  fixed $\hat{\gamma}$, $i\in[m]$, $c\in \bit^{m}$ such that $c_i=0$, $k_i$, $\td_i$, and $y_i$, we have
\begin{align}
 \Pr\left[M_{i,k_i,\td_i,y_i,0} \circ \frac{U\ket{c}_{\regC}\ket{\psi_{\bar{c}_1,\ldots,\bar{c}_{i-1},0}}_{\regX,\regZ}}{\|\ket{\psi_{\bar{c}_1,\ldots,\bar{c}_{i-1},0}}_{\regX,\regZ}\|}=\top\right]\leq 2^{m-1}\gamma_0+\negl(\secpar). \label{eq:Test}
\end{align}
This easily follows from the first claim of Lemma~\ref{lem:partition_further}

Second, for any  fixed $\hat{\gamma}$, $i\in[m]$, and $c\in \bit^{m}$ such that $c_i=1$,
\begin{align}
    \underset{k,\td,y}{E}\left[\|\ket{\psi_{\bar{c}_1,\dots,\bar{c}_{i-1},1}}_{\regX,\regZ}\|^2\Pr\left[M_{i,k_i,\td_i,y_i,1}\circ U\frac{\ket{c}_{\regC}\ket{\psi_{\bar{c}_1,\dots,\bar{c}_{i-1},1}}_{\regX,\regZ}}{\|\ket{\psi_{\bar{c}_1,\dots,\bar{c}_{i-1},1}}_{\regX,\regZ}\|} = \top\right]\right] = \negl(n) \label{eq:Hada}
\end{align}
assuming the quantum hardness of LWE problem.

For proving Eq.~\ref{eq:Hada}, we consider a cheating prover against the original Mahadev's protocol on the $i$-th corrdinate described below:

\begin{enumerate}
    \item Given $k_i$, it picks $k_{-i}=k_1...k_{i-1},k_{i+1},...,k_{m}$ as in the protocol and computes $U_0\ket{0}_{\regX,\regZ}\ket{0}_{\regY}\ket{k}_{\regK}$ and measure the register $\regY$ to obtain $y=(y_1,...,y_m)$ along with the corresponding state $\ket{\psi}_{\regX,\regZ}=\ket{\psi(k,y)}_{\regX,\regZ}$.
    \item Apply $H_{\hat{\gamma},c}$ to generate the state $\frac{\ket{\psi_{\bar{c}_1,\dots,\bar{c}_{i-1},1}}_{\regX,\regZ}}{\|\ket{\psi_{\bar{c}_1,\dots,\bar{c}_{i-1},1}}_{\regX,\regZ}\|}$, which succeeds with probability $\|\ket{\psi_{\bar{c}_1,\dots,\bar{c}_{i-1},1}}_{\regX,\regZ}\|^2$ (ignoring a global phase factor).
    We denote by $\Succ$ the event that it succeeds in generating the state.
    If it fails to generate the state, then it overrides $y_i$ by picking it in a way such that it can pass the test round with probability $1$, which can be done according to Fact~\ref{fact:perfectly_pass_test}.
    Then it sends $y_i$ to the verifier.
    \item Given a challenge $c'_i$, it works as follows:
    \begin{itemize}
     \item When $c'_i=0$ (i.e., Test round), if $\Succ$ occurred, then it runs $\ext_i$ in the second claim of Lemma~\ref{lem:partition_further} on input $\frac{\ket{0^m}_{\regC}\ket{\psi_{\bar{c}_1,\dots,\bar{c}_{i-1},1}}_{\regX,\regZ}}{\|\ket{\psi_{\bar{c}_1,\dots,\bar{c}_{i-1},1}}_{\regX,\regZ}\|}$ to generate a forth message accepted with probability $1-\negl(\secpar)$. 
     If $\Succ$ did not occur, then it returns a forth message accepted with probability $1$, which is possible by Fact~\ref{fact:perfectly_pass_test}.
    \item When $c'_i=1$ (i.e., Hadamard round), if $\Succ$ occurred, then it computes  $U\frac{\ket{c}_{\regC}\ket{\psi_{\bar{c}_1,\dots,\bar{c}_{i-1},1}}_{\regX,\regZ}}{\|\ket{\psi_{\bar{c}_1,\dots,\bar{c}_{i-1},1}}_{\regX,\regZ}\|}$ and measure the register $\regX_i$ to obtain the forth message $a_i$.
    If $\Succ$ did not occur, it just aborts.
    \end{itemize}
\end{enumerate}
Then we can see that this cheating adversary passes the test round with overwhelming probability and passes the Hadamard round with the probability equal to the LHS of Eq.~\ref{eq:Hada}.
Therefore, Eq.~\ref{eq:Hada} follows from Lemma~\ref{lem:Mah_soundness} assuming the quantum hardness of LWE problem.

Now, we are ready to prove the theorem. 
As remarked at the beginning of Sec. \ref{sec:proof_of_soundness}, it suffices to show that for any $\mu=1/\poly(n)$, there exists $m=O(\log(n))$ such that the success probability of the cheating prover is at most $\mu$.
Here we set $m = \log \frac{1}{\mu^2}$, $\gamma_0 = 2^{-2m}$, and $T=2^{m}$. 
Note that this parameter setting satisfies the requirement for Lemma~\ref{lem:partition_further} since
$m=\log \frac{1}{\mu^2}=\log (\poly(\secpar))=O(\log \secpar)$ and
$\frac{\gamma_0}{T}=2^{-3m}=\mu^{6}=1/\poly(\secpar)$.
Then we have
\begin{align*}
    &\Pr_{k,\td,y,c}\left[M_{k,\td,y,c}\circ U\ket{c}_{\regC}\ket{\psi}_{\regX,\regZ}= \top\right] \\
     &=\Pr_{k,\td,y,c,\hat{\gamma}}\left[M_{k,\td,y,c}\circ U\ket{c}_{\regC}\left(\sum_{i=1}^{m}\ket{\psi_{\bar{c}_1,\dots,\bar{c}_{i-1},c_i}}_{\regX,\regZ} + \ket{\psi_{\bar{c}_1,\dots,\bar{c}_m}}_{\regX,\regZ}+ \ket{\psi_{err}}_{\regX,\regZ}\right) = \top\right] \\
    &\leq (m+2) \underset{k,\td,y,c,\hat{\gamma}}{E}\Biggl[\sum_{i=1}^{m} \|\ket{\psi_{\bar{c}_1,\dots,\bar{c}_{i-1},c_i}}_{\regX,\regZ}\|^2\Pr\left[M_{k,\td,y,c}\circ U\frac{\ket{c}_{\regC}\ket{\psi_{\bar{c}_1,\dots,\bar{c}_{i-1},c_i}}_{\regX,\regZ}}{\|\ket{\psi_{\bar{c}_1,\dots,\bar{c}_{i-1},c_i}}_{\regX,\regZ}\|}=\top\right]\\
   &+\|\ket{\psi_{\bar{c}_1,\dots,\bar{c}_m}}_{\regX,\regZ}\|^2\Pr\left[M_{k,\td,y,c}\circ U\frac{\ket{c}_{\regC}\ket{\psi_{\bar{c}_1,\dots,\bar{c}_m}}_{\regX,\regZ}}{\|\ket{\psi_{\bar{c}_1,\dots,\bar{c}_m}}_{\regX,\regZ}} =\top\right]\\
    &+ \|\ket{\psi_{err}}_{\regX,\regZ}\|^2\Pr\left[M_{k,\td,y,c}\circ U\frac{\ket{c}_{\regC} \ket{\psi_{err}}_{\regX,\regZ}}{\|\ket{\psi_{err}}_{\regX,\regZ}\|}=\top\right]\Biggr]\\
        &\leq (m+2) \underset{k,\td,y,c,\hat{\gamma}}{E}\Biggl[\sum_{i=1}^{m} \|\ket{\psi_{\bar{c}_1,\dots,\bar{c}_{i-1},c_i}}_{\regX,\regZ}\|^2\Pr\left[M_{i,k_i,\td_i,y_i,c_i}\circ U\frac{\ket{c}_{\regC}\ket{\psi_{\bar{c}_1,\dots,\bar{c}_{i-1},c_i}}_{\regX,\regZ}}{\|\ket{\psi_{\bar{c}_1,\dots,\bar{c}_{i-1},c_i}}_{\regX,\regZ}\|}=\top\right]\\
   &+\|\ket{\psi_{\bar{c}_1,\dots,\bar{c}_m}}_{\regX,\regZ}\|^2+ \|\ket{\psi_{err}}_{\regX,\regZ}\|^2 \Biggr]\\
    &\leq (m+2)(m(2^{m-1}\gamma_0 +\negl(n))+ 2^{-m} + \frac{m^2}{T}+\negl(\secpar)) \\
    & \leq \mathsf{poly}(\log \mu^{-1}) \mu^{2}+\negl(\secpar). 
\end{align*}
The first equation follows from Lemma \ref{lem:partition_further}. The first inequality follows from Lemma~\ref{lem:Cauchy-Schwarz}.
 The second inequality holds since considering the verification on a particular coordinate just increases the acceptance probability and probabilities are at most $1$.
The third inequality follows from Eq.~\ref{eq:Test} and \ref{eq:Hada}, which give an upper bound of the first term and Lemma~\ref{lem:partition_further}, which gives upper bounds of the second and third terms.
The last inequality follows from our choices of $\gamma_0$, $T$, and $m$.
For sufficiently large $\secpar$,  this can be upper bounded by $\mu$.
\end{proof}

\section{Two-Round Protocol via Fiat-Shamir Transform}\label{sec:tworound}
In this section, we show that if we apply the Fiat-Shamir transform to $m$-parallel version of the Mahadev's protocol, then we obtain two-round protocol in the QROM. 
That is, we prove the following theorem.
\begin{theorem}\label{thm:MahFS}
Assuming LWE assumption, there exists a two-round CVQC protocol with overwhelming completeness and negligible soundness error in the QROM.
\end{theorem}

\begin{proof}
Let $m>\secpar$ be a sufficiently large integer so that $m$-parallel version of the Mahadev's protocol has negligible soundness.
For notational simplicity, we abuse the notation to simply use $\ver_{i}$, $\pro_{i}$, and $\ver_{\out}$ to mean the $m$-parallel repetitions of them. 
Let $H:\calY\ra \bit^{m}$ be a hash function idealized as a quantum random oracle where $\calX$ is the space of the second message $y$ and $\calY=\bit^{m}$.
Our two-round protocol is described below:
\begin{description}
\item[First Message:] The verifier runs $\ver_1$ to generate $(k,\td)$. Then it sends $k$ to the prover and keeps $\td$ as its state.
\item[Second Message:] The prover runs $\pro_2$ on input $k$ to generate $y$ along with the prover's state $\ket{\st_\pro}$. Then set $c\defeq H(y)$, and runs $\pro_4$ on input $\ket{\st_\pro}$ and $y$ to generate $a$. Finally, it returns $(y,a)$ to the verifier.
\item[Verification:] The verifier computes $c=H(y)$, runs $\ver_{\out}(k,\td,y,c,a)$, and outputs as $\ver_{\out}$ outputs.
\end{description}

It is clear that the completeness is preserved given that $H$ is a random oracle.
\ifnum\submission=1
We can reduce the soundness of this protocol to the soundness of $m$-parallel version of the Mahadev's protocol by using the result of \cite{C:DFMS19}, which shows that Fiat-Shamir transform preserves soundness in the QROM.
See Appendix \ref{app:fiat-shamir} for details.
\else
We reduce the soundness of this protocol to the soundness of $m$-parallel version of the Mahadev's protocol.
For proving this, we borrow the following lemma shown in \cite{C:DFMS19}.


\begin{lemma}[{\cite[Theorem 2]{C:DFMS19}}]\label{lem:FS}
Let $\calY$ be finite non-empty sets. There exists a black-box polynomial-time two-stage quantum algorithm $\calS$ with the following property. Let $\A$ be an arbitrary oracle quantum algorithm that makes $q$ queries to a uniformly random $H:\calY\ra \bit^{m}$ and that outputs some $y\in \calY$ and output $a$. 
Then, the two-stage algorithm $\calS^{\A}$ outputs $y\in\calY$ in the first stage
and, upon a random $c\in \bit^{m}$ as input to the second stage, output $a$ so that for any $x_\circ\in \calX$ and any predicate $V$:
\begin{align*}
    \Pr_{c}\left[V(y,c,a):(y,a)\sample \langle\calS^{\A},c \rangle \right]\leq \frac{1}{O(q^2)}\Pr_{H}\left[V(y,H(y),a):(y,a)\sample \A^H\right]-\frac{1}{2^{m+1}q},
\end{align*}
where 
$(y,a)\sample \langle\calS^{\A},c \rangle$ means that $\calS^{\A}$ outputs $y$ and $a$ in the first and second stages respectively on the second stage input $c$.
\end{lemma}

We assume that there exists an efficient adversary $\A$ that breaks the soundness of the above two-round protocol.
We fix $x\notin \lang$ on which $\A$ succeeds in cheating.
We fix $(k,\td)$ that is in the support of the verifier's first message.
We apply Lemma~\ref{lem:FS} for $\A=\A(k)$ and $V=\ver_\out(k,\td,\cdot,\cdot,\cdot)$, to obtain an algorithm $\calS^{\A(k)}$ that satisfies
\begin{align*}
    &\Pr_{c}\left[V_{\out}(k,\td,y,c,a):(y,a)\sample \langle\calS^{\A(k)},c \rangle \right]\\
    \leq &\frac{1}{O(q^2)}\Pr_{H}\left[V_{\out}(k,\td,y,H(y),a):(y,a)\sample \A^H(k)\right]-\frac{1}{2^{m+1}q}.
\end{align*}
Averaging over all possible $(k,\td)$, we have
\begin{align*}
    &\Pr_{k,\td,c}\left[V_{\out}(k,\td,y,c,a):(y,a)\sample \langle\calS^{\A(k)},c \rangle \right]\\
    \leq &\frac{1}{O(q^2)}\Pr_{k,\td,H}\left[V_{\out}(k,\td,y,H(y),a):(y,a)\sample \A^H(k)\right]-\frac{1}{2^{m+1}q}.
\end{align*}
Since we assume that $\A$ breaks the soundness of the above two-round protocol,
\[
\Pr_{k,\td,H}\left[V_{\out}(k,\td,y,H(y),a):(y,a)\sample \A^H(k)\right]
\]
is non-negligible in $\secpar$.
Therefore, as long as $q=\poly(\secpar)$, 
\[
\Pr_{k,\td,c^*}\left[V_{\out}(k,\td,y,c^*,a):(y,a)\sample \langle\calS^{\A(k)},c^* \rangle \right]\]
is also non-negligible in $\secpar$.
Then, we construct an adversary $\B$ that breaks the soundness of parallel version of Mahadev's protocol as follows:
\begin{description}
\item[Second Message:] Given the first message $k$, $\B$ runs the first stage of $\calS^{\A(k)}$ to obtain $y$. It sends $y$ to the verifier.
\item[Forth Message:] Given the third message $c$, $\B$ gives $c$ to $\calS^{\A(k)}$ as the second stage input, and let $a$ be the output of it.
Then $\B$ sends $a$ to the verifier.
\end{description}
Clearly, the probability that $\B$ succeeds in cheating is 
\[
\Pr_{k,\td,c^*}\left[V_{\out}(k,\td,y,c^*,a):(y,a)\sample \langle\calS^{\A(k)},c^* \rangle \right],\]
which is non-negligible in $\secpar$.
This contradicts the soundness of $m$-parallel version of Mahadev's protocol (Theorem~\ref{thm:rep_soundness}).
Therefore we conclude that there does not exists an adversary that succeeds in the two-round protocol with non-negligible probability assuming LWE in the QROM.
\fi
\end{proof}

\section{Making Verifier Efficient}\label{sec:efficient}
In this section, we construct a CVQC protocol with efficient verification in the CRS+QRO model where a classical common reference string is available for both prover and verifier in addition to quantum access to QRO.
Our main theorem in this section is stated as follows:
\begin{theorem}\label{thm:Eff}
Assuming LWE assumption and existence of post-quantum iO, post-quantum FHE, and two-round CVQC protocol in the standard model, there exists a two-round CVQC protocol for $\QTIME(T)$ with verification complexity $\poly(n, \log T)$ in the CRS+QRO model.
\end{theorem}
\begin{remark}
One may think that the underlying two-round CVQC protocol can be in the QROM instead of in the standard model since we rely on the QROM anyway.
However, this is not the case since we need to use the underlying two-round CVQC in a non-black box way, which cannot be done if that is in the QROM.
Since our two-round protocol given in Sec.~\ref{sec:tworound} is only proven secure in the QROM, we do not know any two-round CVQC protocol provably secure in the standard model.
On the other hand, it is widely used heuristic in cryptography that a scheme proven secure in the QROM is also secure in the standard model if the QRO is instantiated by a well-designed cryptographic hash function such as SHA-3. 
Therefore, we believe that it is reasonable to assume that a standard model instantiation of the scheme in  Sec.~\ref{sec:tworound} with a concrete hash function is sound.  
\end{remark}
\begin{remark}
One may think we need not assume CRS in addition to QRO since CRS may be replaced with an output of QRO.
This can be done if CRS is just a uniformly random string.
However, in our construction, CRS is non-uniform and has a certain structure.
Therefore we cannot implement CRS by QRO.
\end{remark}

\subsection{Preparation}
First, we prepare a lemma that is used in our security proof.
\begin{lemma}\label{lem:adaptive_program}
For any finite sets $\calX$ and $\calY$ and two-stage oracle-aided quantum algorithm $\A=(\A_1,\A_2)$, we have
\begin{align*}
\Pr\left[1 \sample \A_2^{H}(\ket{\st_\A},z):\ket{\st_\A}\sample\A_1^{H}()\right]-\Pr\left[1 \sample \A_2^{H[z,G]}(\ket{\st_\A},z):\ket{\st_\A}\sample\A_1^{H}()\right]\leq q_12^{-\frac{\ell}{2}+1}
\end{align*}
where 
$z\sample \bit^{\ell}$,
$H\sample \func(\bit^{\ell}\times \calX,\calY)$, $G\sample \func(\calX,\calY)$, $H[z,G]$ is defined by 
\begin{align*}
H[z,G](z',x)=
\begin{cases}
G(x) &\text{~if~}z'=z \\
H(z',x) &\text{~else}
\end{cases}.    
\end{align*}
where $q_1$ denotes the maximal number of queries by $\A_1$.
\end{lemma}
This can be proven similarly to \cite[Lemma 2.2]{EC:SaiXagYam18}. 
We give a proof in Appendix~\ref{sec:proof_adaptive_program} for completeness.

\subsection{Four-Round Protocol}\label{sec:efficient-four}
First, we construct a four-round scheme with efficient verification, which is transformed into two-round protocol in the next subsection.
Our construction is based on the following building blocks:
\begin{itemize}
\item A two-round CVQC protocol $\Pi=(\pro=\pro_2,\ver=(\ver_1,\ver_\out))$ in the standard model, which works as follows: 
\begin{description}
\item[$\ver_1$:] On input the security parameter $1^\secpar$ and $x$, it generates a pair $(\key,\td)$ of a``key" and ``trapdoor", sends $\key$ to $\pro$, and keeps $\td$ as its internal state.
\item[$\pro_2$:] On input $x$ and $\key$, it generates a response $e$ and sends it to $\ver$.
\item[$\ver_\out$:] On input $x$, $\key$, $\td$, $e$, it returns $\top$ indicating acceptance or $\bot$ indicating rejection.
\end{description}

\item A post-quantum PRG $\PRG:\bit^{\ell_s}\ra \bit^{\ell_r}$ where $\ell_r$ is the length of randomness for $\ver_1$.


\item An FHE scheme $\Pi_\FHE=(\fhekeygen,\fheenc,\fheeval,\fhedec)$ with post-quantum CPA security.

\item A strong output compressing randomized encoding scheme $\Pi_\RE=(\rsetup,\renc,\rdec)$ with post-quantum security. We denote the simulator for $\Pi_\RE$ by $\calS_\re$.

\item A SNARK $\Pi_{\SNARK}=(\pro_{\snark},\ver_{\snark})$ in the QROM for an $\NP$ language $\lang_{\snark}$ defined below:  

We have $(x,\pk_{\fhe},\ct,\ct')\in \lang_{\snark}$
if and only if there exists $e$ such that
$\ct'= \fheeval(\pk_{\fhe},\allowbreak C[x,e],\ct)$ where $C[x,e]$ is a circuit that works as follows:
\begin{description}
\item[$C{[}x,e{]}(s)$:] Given input $s$, it computes $(k,\td)\sample \ver_1(1^\secpar,x;PRG(s))$, and returns $1$ if and only if $\ver_\out(x,k,\td,e)=\top$ and $0$ otherwise. 
\end{description}

\end{itemize}

Let $\lang$ be a BPP language decided by a quantum Turing machine $\QTM$ (i.e., for any $x\in \bit^*$, $x\in \lang$ if and only if $\QTM$ accepts $x$),
and for any $T$, $\lang_T$ denotes the set consisting of $x\in \lang$ such that 
$\QTM$ accepts $x$ in $T$ steps.
Then we construct a 4-round CVQC protocol $(\setupeff,\proeff=(\proefftwo,\proefffour),\vereff=(\vereffone,\vereffthree,\vereffout))$ for $\lang_T$ in the CRS+QRO model where the verifier's efficiency only logarithmically depends on $T$.
Let $H:\bit^{2\secpar}\times \bit^{2\secpar}\ra \bit^{\secpar}$ be a quantum random oracle. 

\begin{description}
\item[$\setupeff(1^\secpar)$:]
The setup algorithm takes the security parameter $1^\secpar$ as input, generates $\crs_\re \sample \bit^{\ell}$ 
and computes $\ek_\re\sample \rsetup(1^\secpar,1^\ell,\crs_\re)$ 
where $\ell$ is a parameter specified later.
Then it outputs a CRS for verifier $\crs_{\vereff}\defeq \ek_\re$ and a CRS for prover $\crs_{\proeff}\defeq \crs_\re$.\footnote{We note that we divide the CRS into $\crs_{\vereff}$ and $\crs_{\proeff}$ just for the verifier efficiency and soundness still holds even if a cheating prover sees $\crs_{\vereff}$.} 
\item[$\vereffone^{H}$:] 
Given $\crs_{\vereff}= \ek_\re$ and $x$, 
it generates $s\sample \bit^{\ell_s}$ and $(\pk_{\fhe},\sk_{\fhe})\sample \fhekeygen(1^\secpar)$,
computes $\ct\sample \fheenc(\pk_{\fhe},s)$
and  $\Menc\sample \renc(\ek_\re,M,s,T')$
where $M$ is a Turing machine that works as follows:
\begin{description}
\item[$M(s)$:] Given an input $s\in \bit^{\ell_s}$, it computes $(k,\td)\sample \ver_1(1^\secpar,x;PRG(s))$ and outputs $k$
\end{description}
and $T'$ is specified later.
Then it sends $(\Menc,\pk_{\fhe},\ct)$ to $\proeff$ and keeps $\sk_{\fhe}$ as its internal state.

\item[$\proefftwo^{H}$:] Given $\crs_{\proeff}=\crs_\re$, $x$ and the message  $(\Menc,\pk_{\fhe},\ct)$ from the verifier, it computes $k\la \rdec(\crs_{\re},\Menc)$, $e\sample \pro_2(x,k)$, and $\ct'\la \fheeval(\pk_{\fhe},C[x,e],\ct)$ where $C[x,e]$ is a classical circuit defined above.
Then it sends $\ct'$ to $\pro$ and keeps $(\pk_{\fhe},\ct,\ct',e)$ as its state.
\item[$\vereffthree^{H}$] Upon receiving $\ct'$, it randomly picks $z\sample \bit^{2\secpar}$ and sends $z$ to $\proeff$.
\item[$\proefffour^{H}$] Upon receiving $z$, 
it computes $\pi_\snark \sample \pro_{\snark}^{H(z,\cdot)}((x,\pk_{\fhe},\ct,\ct'),e)$
and sends $\pi_\snark$ to $\vereff$.

\item[$\vereffout^{H}$:] 
It returns $\top$ if $\ver_{\snark}^{H(z,\cdot)}((x,\pk_{\fhe},\ct,\ct'),\pi_{\snark})=\top$ and $1 \la \fhedec(\sk_{\fhe},\ct')$ and $\bot$ otherwise.
\end{description}

\paragraph{Choice of parameters.}
\begin{itemize}
\item We set $\ell$ to be an upper bound of the length of $k$ where $(k,\td)\sample \ver_1(1^\secpar,x)$ for $x\in \lang_T$. We note that we have $\ell=\poly(\secpar,T)$.
\item We set $T'$ to be an upperbound of the running time of $M$ on input $s\in \bit^{\ell_s}$ when $x\in \lang_T$. We note that we have $T'=\poly(\secpar,T)$.
\end{itemize}

\paragraph{Verification Efficiency.}
By encoding efficiency of $\Pi_{RE}$ and verification efficiency of $\Pi_{\SNARK}$, $\vereff$ runs in time $\poly(\secpar,|x|,\log T)$.

\begin{theorem}[Completeness]\label{thm:eff_completeness}
For any $x\in \lang_T$, 
\begin{align*}
\Pr\left[\langle\proeff^{H}(\crs_{\proeff}),\vereff^H(\crs_{\vereff})\rangle(x)=\bot \right]=\negl(\secpar)    
\end{align*}
where $(\crs_{\proeff},\crs_{\vereff})\sample \setupeff(1^\secpar)$.
\end{theorem}
\begin{proof}
This easily follows from completeness and correctness of the underlying primitives.
\end{proof}

\begin{theorem}[Soundness]\label{thm:eff_soundness}
For any $x\notin \lang_T$ any efficient quantum cheating prover $\A$, 
\begin{align*}
\Pr\left[\langle\A^H(\crs_{\proeff},\crs_{\vereff}),\vereff^H(\crs_{\vereff})\rangle(x)=\top\right]=\negl(\secpar)    
\end{align*}
where $(\crs_{\proeff},\crs_{\vereff})\sample \setupeff(1^\secpar)$.
\end{theorem}
\ifnum\submission=1
We give a proof in Appendix \ref{app:proof_eff_soundness}
\else
\begin{proof}
We fix $T$ and $x\notin \lang_T$.
Let $\A$ be a cheating prover.
First, we divides $\A$ into the first stage $\A_1$, which is given $(\crs_{\proeff},\crs_{\vereff})$ and the first message and outputs the second message $\ct'$ and its internal state $\ket{\st_{\A}}$, and the second stage $\A_2$, which is given the internal state $\ket{\st_{\A}}$ and the third message and outputs the fourth message $\pi_{\snark}$.
We consider the following sequence of games between an adversary $\A=(\A_1,\A_2)$ and a challenger. 
Let $q_1$ and $q_2$ be an upper bound of number of random oracle queries by $\A_1$ and $\A_2$, respectively.
We denote the event that the challenger returns $1$ in $\game_i$ by $\TT_i$.
\begin{description}
\item[$\game_1$:] This is the original soundness game.
Specifically, the game runs as follows:
\begin{enumerate}
    \item The challenger generates 
    $H\sample \func(\bit^{2\secpar}\times \bit^{2\secpar},\bit^{\secpar})$, 
    $\crs_\re\sample \bit^{\ell}$, 
    $s\sample \bit^{\ell_s}$, and $(\pk_{\fhe},\sk_{\fhe})\sample \fhekeygen(1^\secpar)$, and computes $\ek_\re\sample \rsetup(1^\secpar,1^\ell,\crs_\re)$,  
    $\ct\sample \fheenc(\pk_{\fhe},s)$, and  $\Menc\sample \renc(\ek_\re,M,s,T')$.
    \item $\A_1^{H}$ is given $\crs_{\proeff}\defeq \crs_\re$, $\crs_{\vereff}\defeq \ek_\re$ and the first message $(\Menc,\pk_{\fhe},\ct)$, and outputs the second message $\ct'$ and its internal state $\ket{\st_{\A}}$.
    \item The challenger randomly picks $z\sample \bit^{2\secpar}$.
    \item $\A_2^{H}$ is given the state $\ket{\st_{\A}}$ and the third message $z$ and outputs $\pi_{\snark}$.
    \item The challenger returns $1$ if  $\ver_{\snark}^{H(z,\cdot)}((x,\pk_{\fhe},\ct,\ct'),\pi_{\snark})=\top$ and $1 \la \fhedec(\sk_{\fhe},\ct')$ and $0$ otherwise.
\end{enumerate}

\item[$\game_2$:]
This game is identical to the previous game except that the oracles given to $\A_2$ and $V_{\snark}$ are replaced with $H[z,G]$ and $G$ in Step 4 and 5 respectively where $G\sample \func(\bit^{2\secpar},\bit^{\secpar})$ and $H[z,G]$ is as defined in Lemma~\ref{lem:adaptive_program}. 
We note that the oracle given to $\A_1$ in Step 2 is unchanged from $H$.

\item[$\game_3$:]
This game is identical to the previous game except that Step 4 and 5 are modified as follows:
\begin{description}
\item[4.] The challenger runs $e \sample \ext^{\A'_2[H,\ket{\st_{\A}},z]}((x,\pk_{\fhe},\ct,\ct'),1^{q_2},1^\secpar)$ where $\A'_2[H,\st_{\A},z]$ is an oracle-aided quantum algorithm that is given an oracle $G$ and emulates $\A_2^{H[z,G]}(\ket{\st_{\A}},z)$.
\item[5.] The challenger returns $1$ if 
$e$ is a valid witness for $(x,\pk_{\fhe},\ct,\ct')\in \lang_{\snark}$ and $1 \la \fhedec(\sk_{\fhe},\ct')$ and $0$ otherwise.
\end{description}


\item[$\game_4$:]
This game is identical to the previous game except that Step 5 is modified as follows:
\begin{description}
\item[5.] The challenger returns $1$ if 
$e$ is a valid witness for $(x,\pk_{\fhe},\ct,\ct')\in \lang_{\snark}$, and $\ver_\out(x,k,\td,e)=\top$ where $(k,\td)\sample \ver_1(1^\secpar,x;PRG(s))$ and $0$ otherwise.
\end{description}

\item[$\game_5$:]
This game is identical to the previous game except that $\ct$ is generated as $\ct\sample \fheenc(\pk_{\fhe},\allowbreak 0^{2\secpar})$ in Step 1.

\item[$\game_6$:]
This game is identical to the previous game except that $\crs_\re$, $\ek_\re$, and $\Menc$ are generated in a different way. Specifically, in Step $1$, the challenger computes
$(k,\td)\sample \ver_1(1^\secpar,x;PRG(s))$, $(\crs_\re,\Menc)\sample \calS_\re(1^\secpar,1^{|M|},1^{\ell_s},k,T^*)$, and $\ek_\re \sample \rsetup(1^\secpar,1^\ell,\crs_\re)$ where $T^*$ is the running time of $M(\inp)$.
We note that the same $(k,\td)$ generated in this step is also used in Step 5.

\item[$\game_7$:]
This game is identical to the previous game except that $PRG(s)$ used for generating $(k,\td)$ in Step 1 is replaced with a true randomness.
\end{description}

This completes the descriptions of games. 
Our goal is to prove $\Pr[\TT_1]=\negl(\secpar)$. We prove this by the following lemmas. 
Since Lemmas~\ref{lem:eff_gamehop_five}, \ref{lem:eff_gamehop_six}, and \ref{lem:eff_gamehop_seven} can be proven by straightforward reductions, we only give proofs for the rest of lemmas.

\begin{lemma}\label{lem:eff_gamehop_one}
We have $|\Pr[\TT_2]-\Pr[\TT_1]|\leq q_12^{-(\secpar+1)}$.
\end{lemma}
\begin{proof}
This lemma is obtained by applying Lemma~\ref{lem:adaptive_program} for $\B=(\B_1,\B_2)$ described below:
\begin{description}
\item[$\B_1^{O_1}$():]   It generates 
    $\crs_\re\sample \bit^{\ell}$,  $s\sample \bit^{\ell_s}$, and $(\pk_{\fhe},\sk_{\fhe})\sample \fhekeygen(1^\secpar)$, computes $\ek_\re\sample \rsetup(1^\secpar,1^\ell,\crs_\re)$,  $\ct\sample \fheenc(\pk_{\fhe},s)$, $\Menc\sample \renc(\ek_\re,M,s,T')$,  and $\ct\sample \fheenc(\pk_{\fhe},s)$, and sets $\crs_{\proeff}=\crs_\re$ and $\crs_{\vereff}\defeq \ek_\re$.
Then it runs $(\ct',\ket{\st_{\A}})\sample \A_1^{O_1}(\crs_{\proeff},\crs_{\vereff},x,(\Menc,\pk_{\fhe},\ct))$, and outputs $\ket{\st_\B}\defeq(\ket{\st_\A},x,\Menc,\ct,\ct',\sk_\fhe)$.\footnote{Classical strings are encoded as quantum states in a trivial manner.}
\item[$\B_2^{O_2}(\ket{\st_\B},z)$:] 
It runs $\pi_{\snark}\sample \A_2^{O_2}(\ket{\st_\A})$, and outputs $1$ if $\ver_{\snark}^{O_2(z,\cdot)}((x,\pk_{\fhe},\ct,\ct'),\pi_{\snark})=\top$ and $1 \la \fhedec(\sk_{\fhe},\ct')$ and $0$ otherwise.
\end{description}
\end{proof}

\begin{lemma}\label{lem:eff_gamehop_two}
If $\Pi_{\SNARK}$ satisfies the extractability and $\Pr[\TT_2]$ is non-negligible, then $\Pr[\TT_3]$ is also non-negligible.
\end{lemma}
\begin{proof}
Let $\transcript_3$ be the transcript of the protocol before the forth message is sent (i.e., $\transcript_3=(\crs_{\proeff},\crs_{\vereff},\Menc,\pk_{\fhe},\ct',z)$).
We say that $(H,\sk_\fhe,\transcript_3,\ket{\st_\A})$ is good if we randomly choose $G\sample \func(\bit^{2\secpar},\bit^{\secpar})$ and run $\pi_{\snark}\sample \A_2^{H[z,G]}(\ket{\st_\A})$ to complete the transcript, then the transcript is accepted (i.e., we have $\ver_{\snark}^{G}((x,\pk_{\fhe},\ct,\ct'),\pi_{\snark})=\top$ and $1 \la \fhedec(\sk_{\fhe},\ct')$) with non-negligible probability.
By a standard averaging argument, if $\Pr[\TT_2]$ is non-negligible, then a non-negligible fraction of $(H,\sk_\fhe,\transcript_3,\ket{\st_\A})$ is good when they are generated as in $\game_2$.
We fix good $(\transcript_3,\sk_\fhe,\ket{\st_\A})$.
Then by the extractability of $\Pi_{\SNARK}$, $\ext$ succeeds in extracting a witness for $(x,\pk_{\fhe},\ct,\ct')\in \lang_{\snark}$ with non-negligible probability. Moreover, since we assume $(H,\sk_\fhe,\transcript_3,\ket{\st_\A})$ is good, we always have $1 \la \fhedec(\sk_{\fhe},\ct')$ (since otherwise a transcript with prefix $\transcript_3$ cannot be accepted).
Therefore we can conclude that $\Pr[\TT_3]$ is non-negligible.
\end{proof}


\begin{lemma}\label{lem:eff_gamehop_four}
We have $\Pr[\TT_4]=\Pr[\TT_3]$.
\end{lemma}
\begin{proof}
If $e$ is a valid witness for $(x,\pk_{\fhe}, \ct,\ct')\in \lang_{\snark}$, then we especially have $\ct'= \fheeval(\pk_{\fhe},\allowbreak C[x,e],\ct)$.
By the correctness of $\Pi_\FHE$, we have $\fhedec(\sk_\fhe,\ct')=C[x,e](s)=(\ver_\out(x,k,\td,e)\overset{?}{=}\top)$ where $(k,\td)\sample \ver_1(1^\secpar,x;PRG(s))$.
Therefore, the challenger returns $1$ in $\game_4$ if and only if it returns $1$ in $\game_3$.
\end{proof}

\begin{lemma}\label{lem:eff_gamehop_five}
If $\Pi_{\FHE}$ is CPA-secure, then we have $|\Pr[\TT_5]-\Pr[\TT_4]|\leq \negl(\secpar)$.
\end{lemma}

\begin{lemma}\label{lem:eff_gamehop_six}
If $\Pi_{\RE}$ is secure, then we have
$|\Pr[\TT_6]-\Pr[\TT_5]|\leq \negl(\secpar)$.
\end{lemma}

\begin{lemma}\label{lem:eff_gamehop_seven}
If $\PRG$ is secure, then we have  $|\Pr[\TT_7]-\Pr[\TT_6]|\leq \negl(\secpar)$.
\end{lemma}

\begin{lemma}\label{lem:eff_gamehop_eight}
If $(\pro,\ver)$ satisfies soundness, then we have $\Pr[\TT_7]\leq \negl(\secpar)$.
\end{lemma}
\begin{proof}
Suppose that $\Pr[\TT_7]$ is non-negligible. Then we construct an adversary $\B$ against the underlying two-round protocol as follows:
\begin{description}
\item[$\B(k)$:] Given the first message $k$, it generates 
    $H\sample\func(\bit^{2\secpar}\times \bit^{2\secpar},\bit^{\secpar})$, $G\sample \func(\bit^{2\secpar},\allowbreak \bit^{\secpar})$, $z\sample \bit^{2\secpar}$,
    $(k,\td)\sample \ver_1(1^\secpar,x;PRG(s))$, $(\crs_\re,\Menc)\sample \calS_\re(1^\secpar,1^{|M|},1^{\ell_s},k,T^*)$, $\ek_\re \sample \rsetup(1^\secpar,1^\ell,\crs_\re)$,   and $(\pk_{\fhe},\sk_{\fhe})\sample \fhekeygen(1^\secpar)$,  computes  $\ct\sample \fheenc(\allowbreak \pk_{\fhe},0^{2\secpar})$,
    and sets $\crs_{\proeff}=\crs_\re$ and $\crs_{\vereff}\defeq \ek_\re$.
Then it runs $(\ct',\ket{\st_{\A}})\sample \A_1^{H}(\crs_{\proeff},\crs_{\vereff},\allowbreak x,(\Menc,\pk_{\fhe},\ct))$ and $e\sample \ext^{\A'_2[H,\ket{\st_{\A}},z]}((x,\pk_{\fhe},\ct,\ct'),1^{q_2},1^\secpar)$ and outputs $e$.
\end{description}
Then we can easily see that the probability that we have $\ver_\out(x,k,\td,e)$ is at least $\Pr[\TT_7]$.
Therefore, if the underlying two-round protocol is sound, then $\Pr[\TT_7]=\negl(\secpar)$.
\end{proof}

By combining Lemmas~\ref{lem:eff_gamehop_one} to~\ref{lem:eff_gamehop_seven}, we can see that if $\Pr[\TT_1]$ is non-negligible, then $\Pr[\TT_7]$ is also non-negligible, which contradicts Lemma~\ref{lem:eff_gamehop_eight}.
Therefore we conclude that $\Pr[\TT_1]=\negl(\secpar)$.
\end{proof}
\fi

\subsection{Reducing to Two-Round via Fiat-Shamir}
\ifnum\submission=1
Since the third message is public-coin in the four-round protocol in the previous section, we can  apply the Fiat-Shamir transform similarly to Sec.\ref{sec:tworound}.
Then we obtain the two-round CVQC protocol in the QROM, which completes the proof of Theorem \ref{thm:Eff}.
See Appendix \ref{sec:omitted_two-round_eff} for the full proof.
\else
Here, we show that the number of rounds can be reduced to $2$ relying on another random oracle.
Namely, we observe that the third message of the scheme is just a public coin, and so we can apply the Fiat-Shamir transform similarly to Sec.\ref{sec:tworound}.
In the following, we describe the protocol for completeness.

Our two-round CVQC protocol $(\setupefffs,\proefffs,\verefffs=(\verefffsone,\verefffsout))$ for $\lang_T$ in the CRS+QRO model is described as follows.
Let $H:\bit^{2\secpar}\times \bit^{2\secpar}\ra \bit^{\secpar}$ be a quantum random oracle and $H':\bit^{\ell_{ct'}}\ra \bit^{2\secpar}$ be another quantum random oracle where $\ell_{\ct'}$ is the maximal length of $\ct'$ in the four-round scheme and $\ell$ and $T'$ be as defined in the previous section.

\begin{description}
\item[$\setupefffs(1^\secpar)$:]
The setup algorithm takes the security parameter $1^\secpar$ as input, generates $\crs_\re \sample \bit^{\ell}$ 
and computes $\ek_\re\sample \rsetup(1^\secpar,1^\ell,\crs_\re)$. 
Then it outputs a CRS for verifier $\crs_{\verefffs}\defeq \ek_\re$ and a CRS for prover $\crs_{\proefffs}\defeq \crs_\re$.
\item[$\verefffsone^{H,H'}$:] 
Given $\crs_{\verefffs}= \ek_\re$ and $x$, 
it generates $s\sample \bit^{\ell_s}$ and $(\pk_{\fhe},\sk_{\fhe})\sample \fhekeygen(1^\secpar)$,
computes $\ct\sample \fheenc(\pk_{\fhe},s)$
and  $\Menc\sample \renc(\ek_\re,M,s,T')$
where $M$ is a Turing machine that works as follows:
\begin{description}
\item[$M(s)$:] Given an input $s\in \bit^{\ell_s}$, it computes $(k,\td)\sample \ver_1(1^\secpar,x;PRG(s))$ and outputs $k$.
\end{description}
Then it sends $(\Menc,\pk_{\fhe},\ct)$ to $\proefffs$ and keeps $\sk_{\fhe}$ as its internal state.

\item[$\proefffstwo^{H,H'}$:] Given $\crs_{\proefffs}=\crs_\re$, $x$ and the message  $(\Menc,\pk_{\fhe},\ct)$ from the verifier, it computes $k\la \rdec(\crs_{\re},\Menc)$, $e\sample \pro_2(x,k)$, and $\ct'\la \fheeval(\pk_{\fhe},C[x,e],\ct)$ where $C[x,e]$ is a classical circuit defined above.
Then it computes $z\defeq H'(\ct')$, computes $\pi_\snark \sample \pro_{\snark}^{H(z,\cdot)}((x,\pk_{\fhe},\ct,\ct'),e)$
and sends $(\ct',\pi_\snark)$ to $\verefffs$.

\item[$\verefffsout^{H,H'}$:] 
It computes $z\defeq H'(\ct')$ and returns $\top$ if $\ver_{\snark}^{H(z,\cdot)}((x,\pk_{\fhe},\ct,\ct'),\pi_{\snark})=\top$ and $1 \la \fhedec(\sk_{\fhe},\ct')$ and $\bot$ otherwise.
\end{description}

\paragraph{Verification Efficiency.}
Clearly, the verification efficiency is preserved from the protocol in Sec.~\ref{sec:efficient-four}

\begin{theorem}[Completeness]\label{thm:efffs_completeness}
For any $x\in \lang_T$, 
\begin{align*}
\Pr\left[\langle\proefffs^{H,H'}(\crs_{\proefffs}),\verefffs^{H,H'}(\crs_{\verefffs})\rangle(x)=\bot \right]=\negl(\secpar)    
\end{align*}
where $(\crs_{\proefffs},\crs_{\verefffs})\sample \setupefffs(1^\secpar)$.
\end{theorem}

\begin{theorem}[Soundness]\label{thm:efffs_soundness}
For any $x\notin \lang_T$ any efficient quantum cheating prover $\A$, 
\begin{align*}
\Pr\left[\langle\A^{H,H'}(\crs_{\proefffs},\crs_{\vereff}),\verefffs^{H,H'}(\crs_{\verefffs})\rangle(x)=\top\right]=\negl(\secpar)    
\end{align*}
where $(\crs_{\proefffs},\crs_{\verefffs})\sample \setupefffs(1^\secpar)$.
\end{theorem}
This can be reduced to Theorem~\ref{thm:eff_soundness} similarly to the proof of soundness of the protocol in Sec.~\ref{sec:tworound}.
\fi


\bibliographystyle{alpha}
\bibliography{abbrev3,crypto,reference}

\appendix
\ifnum\submission=0
\section{Proof of Lemma~\ref{lem:adaptive_program}}\label{sec:proof_adaptive_program}
Here, we give a proof of Lemma~\ref{lem:adaptive_program}.
We note that the proof is essentially the same as the proof of \cite[Lemma 2.2]{EC:SaiXagYam18}.

Before proving the lemma, we introduce another lemma, which gives a lower bound for a decisional variant of Grover's search problem. 

\begin{lemma}[{\cite[Lemma C.1]{SonYun17}}] \label{lem:Decision_Grover}
Let $g_z:\bit^\ell \rightarrow \bit$ denotes a function defined as $g_{z}(z):=1$ and $g_z(z'):=0$ for all $z'\not=z$, and $g_{\bot}:\bit^\ell \rightarrow \bit$ denotes a function that returns $0$ for all inputs. Then for any quantum adversary $\B=(\B_1,\B_2)$ we have 
\[
\left|\Pr[1\sample\B_2(\ket{\st_{\B}},z) \mid \ket{\st_{\B}}\sample \B_1^{g_z}()]- \Pr[1\sample\B_2(\ket{\st_{\B}},z) \mid \ket{\st_{\B}}\sample \B_1^{g_{\bot}}()] \right|\leq q_1 \cdot2^{-\frac{\ell}{2}+1}.
\]
where $z\sample \bit^{\ell}$ and $q_1$ denotes the maximal number of queries by $\B_1$.
\end{lemma}

Then we prove Lemma~\ref{lem:adaptive_program}.

\begin{proof}(of Lemma~\ref{lem:adaptive_program}.)
We consider the following sequence of games.
We denote the event that $\game_i$ returns $1$ by $\TT_i$.
\begin{description}
\item[$\game_1$:] This game simulates the environment of the first term of LHS in the inequality in the lemma. Namely, the challenger chooses $z\sample \bit^{\ell}$, $H\sample \func(\bit^{\ell}\times \calX,\calY)$, $\A_1$ runs with oracle $H$ to generate $\ket{\st_{\A}}$, $\A_2$ runs on input $(\ket{\st_{\A}},z)$ with oracle $H$ to generate a bit $b$, and the game returns $b$.
\item[$\game_2$:] This game is identical to the previous game except that the oracle given to $\A_1$ is replaced with $H[z,G]$ where $G\sample \func(\calX,\calY)$.
\item[$\game_3$:] This game is identical to the previous game except that the oracle given to $\A_1$ is replaced with $H$ and the oracle given to $\A_2$ is replaced with $H[z,G]$.
We note that this game simulates the environment as in the second term of the LHS in the inequality in the lemma.
\end{description}
What we need to prove is $|\Pr[\TT_1]-\Pr[\TT_3]|\leq q_12^{-\frac{\ell}{2}+1}$.
First we observe that the change from $\game_2$ to $\game_3$ is just conceptual and nothing changes from the adversary's view since in both games, the oracles given to $\A_1$ and $\A_2$ are random oracles that agrees on any input $(z',x)$ such that $z'\neq z$ and independent on any input $(z,x)$.
Therefore we have $\Pr[\TT_2]=\Pr[\TT_3]$.
What is left is to prove $|\Pr[\TT_1]-\Pr[\TT_2]|\leq q_12^{-\frac{\ell}{2}+1}$.
For proving this, we construct an algorithm $\B=(\B_1,\B_2)$ that breaks Lemma~\ref{lem:adaptive_program} with the advantage $|\Pr[\TT_1]-\Pr[\TT_2]$ as follows:

\begin{description}
\item[$\B_1^{g^*}()$:] It generates $H\sample \func(\bit^{\ell}\times \calX,\calY)$ and  $G\sample \func(\calX,\calY)$, implements an oracle $O_1$ as
\begin{align*}
O_1(z',x)=
\begin{cases}
G(x) &\text{~if~}g^*(z')=1 \\
H(z',x) &\text{~else}
\end{cases},    
\end{align*}
runs $\ket{\st_{\A}}\sample\A_1^{O_1}()$ and outputs $\ket{\st_{\B}}\defeq \ket{\st_{\A}}$
\item[$\B_2(\ket{\st_{\B}}=\ket{\st_{\A}},z)$:]
It runs $b \sample \A_2^{H}(\ket{\st_{\B}},z)$ and outputs $b$.
\end{description}
It is easy to see that if $g^*=g_{\bot}$, then $\B$ perfectly simulates $\game_1$ for $\A$ and if $g^*=g_z$, then $\B$ perfectly simulates $\game_2$ for $\A$.
Therefore, we have $|\Pr[\TT_1]-\Pr[\TT_2]|\leq q_12^{-\frac{\ell}{2}+1}$ by Lemma~\ref{lem:adaptive_program}.
\end{proof} 
\else
\newpage
 	\setcounter{page}{1}
 	{
	\noindent
 	\begin{center}
	{\Large SUPPLEMENTAL MATERIALS}
	\end{center}
 	}
	\setcounter{tocdepth}{2}
\section{Cryptographic Primitives}\label{app:cryptofgraphic_primitives}
\subsubsection{Pseudorandom Generator}
A post-quantum pseudorandom generator (PRG) is an efficient deterministic classical algorithm $\PRG:\bit^{\ell} \ra \bit^{m}$ such that for any efficient quantum algorithm $\A$, we have
\begin{align*}
    \left|\Pr_{s\sample \bit^{\ell}}[\A(\PRG(s))]-\Pr_{y\sample \bit^{m}}[\A(y)]\right|\leq \negl(\secpar).
\end{align*}

It is known that there exists a post-quantum PRG for any $\ell=\Omega(\secpar)$ and $m=\poly(\secpar)$ assuming post-quantum one-way function \cite{SIAM:HILL99,FOCS:Zhandry12}.
Especially, a post-quantum PRG exists assuming the quantum hardness of LWE.



\subsubsection{Fully Homomorphic Encryption}
A post-quantum fully homomorphic encryption consists of four efficient classical algorithm $\Pi_\FHE=\allowbreak (\fhekeygen,\fheenc,\fheeval,\fhedec)$.
\begin{description}
\item[$\fhekeygen(1^\secpar)$:] The key generation algorithm takes the security parameter $1^\secpar$ as input and outputs a public key $\pk$ and a secret key $\sk$.
\item[$\fheenc(\pk,m)$:] The encryption algorithm takes a public key $\pk$ and a message $m$ as input, and outputs a ciphertext $\ct$.
\item[$\fheeval(\pk,C,\ct)$:] The evaluation algorithm takes a public key $\pk$, a classical circuit $C$, and a ciphertext $\ct$, and outputs a evaluated ciphertext $\ct'$.
\item[$\fhedec(\sk,\ct)$:] The decryption algorithm takes secret key $\sk$ and a ciphertext $\ct$ as input and outputs a message $m$ or $\bot$. 
\end{description}
\paragraph{Correctness.}
For all $\secpar\in \mathbb{N}$, $(\pk,\sk)\sample \fhekeygen(1^\secpar)$, $m$ and $C$, we have
\begin{align*}
    \Pr[\fhedec(\sk,\fheenc(\pk,m))=m]=1
\end{align*}
and
\begin{align*}
    \Pr[\fhedec(\sk,\fheeval(\pk,C,\fheenc(\pk,m)))=C(m)]=1.
\end{align*}
\paragraph{Post-Quantum CPA-Security.}
For any efficient quantum adversary $\A=(\A_1,\A_2)$, we have  
\begin{align*}
    &|\Pr[1\sample  \A_2(\ket{\st_{\A}},\ct):(\pk,\sk)\sample \fhekeygen(1^\secpar),(m_0,m_1,\ket{\st_{\A}})\sample \A_1(\pk), \ct\sample \fheenc(\pk,m_0)]\\
    &-\Pr[1\sample  \A_2(\ket{\st_{\A}},\ct):(\pk,\sk)\sample \fhekeygen(1^\secpar),(m_0,m_1,\ket{\st_{\A}})\sample \A_1(\pk), \ct\sample \fheenc(\pk,m_1)]|\\
&\leq \negl(\secpar).
\end{align*}

FHE is usually constructed by first constructing \textit{leveled} FHE, where we have to upper bound the depth of a circuit to evaluate at the setup, and then converting it to FHE by the technique called bootstrapping~\cite{STOC:Gentry09}. 
There have been many constructions of leveled FHE whose (post-quantum) security can be reduced to the (quantum) hardness of LWE \cite{FOCS:BraVai11,ITCS:BraGenVai12,C:Brakerski12,C:GenSahWat13}.
FHE can be obtained assuming that any of these schemes is \textit{circular secure} \cite{EC:CamLys01} so that it can be upgraded into FHE via bootstrapping.
We note that Canetti et al. \cite{TCC:CLTV15} gave an alternative transformation from leveled FHE to FHE based on subexponentially secure iO.

\subsubsection{Strong Output-Compressing Randomized Encoding}
A strong output-compressing randomized encoding \cite{AC:BFKSW19} consists of three efficient classical  algorithms $(\rsetup,\renc,\rdec)$.
\begin{description}
\item[$\rsetup(1^\secpar,1^\ell,\crs)$]: It takes the security parameter $1^\secpar$, output-bound $\ell$, and a common reference string $\crs\in\bit^{\ell}$ and outputs a encoding key $\ek$. 
\item[$\renc(\ek,M,\inp,T)$:] It takes an encoding key $\ek$, Turing machine $M$, an input $\inp\in \bit^*$, and a time-bound $T\leq 2^{\secpar}$ (in binary) as input and outputs an encoding $\Menc$.
\item[$\rdec(\crs,\Menc)$:] It takes a common reference string $\crs$ and an encoding $\Menc$ as input and outputs $\out \in \bit^* \cup \{\bot\}$.
\end{description}

\paragraph{Correctness.}
For any $\secpar\in\mathbb{N}$, $\ell,T\in \mathbb{N}$, $\crs\in \bit^{\ell}$, Turing machine $M$ and input $\inp\in \bit^*$ such that $M(\inp)$ halts in at most $T$ steps and returns a string whose length is at most $\ell$, we have 
\begin{align*}
    \Pr\left[\rdec(\Menc,\crs)=M(\inp): \ek\sample\rsetup(1^\secpar,1^\ell,\crs), \Menc\sample\renc(\ek,M,\inp,T)\right]=1.
\end{align*}

\paragraph{Efficiency.}
There exists polynomials $p_1,p_2,p_3$ such that for all $\secpar\in \mathbb{N}$, $\ell\in \mathbb{N}$, $\crs\sample \bit^{\ell}$:
\begin{itemize}
    \item If $\ek\sample \rsetup(1^\secpar,1^\ell,\crs)$, $|\ek|\leq p_1(\secpar, \log \ell)$.
    \item For every Turing machine $M$, time bound $T$, input $\inp\in \bit^{*}$, if  $\Menc\sample\renc(\ek,M,\inp,T)$, then $|\Menc|\leq p_2(|M|,|\inp|,\log T, \log \ell, \secpar)$,
    \item The running time of $\rdec(\crs,\Menc)$ is at most $\min(T,\Time(M,x))\cdot p_3(\secpar,\log T)$
\end{itemize}

\paragraph{Post-Quantum Security.}

There exists a simulator $\calS$ such that for any $M$ and $\inp$ such that $M(\inp)$ halts in $T^*\leq T$ steps and $|M(\inp)|\leq \ell$ and efficient quantum adversary $\A$,
\begin{align*}
    &|\Pr[1 \sample \A(\crs,\ek,\Menc):\crs\sample \bit^{\ell},\ek\sample\rsetup(1^\secpar,1^\ell,\crs),\Menc\sample\renc(\ek,M,\inp,T)]\\
    &-\Pr[1 \sample \A(\crs,\ek,\Menc):(\crs,\Menc)\sample \calS(1^{\secpar},1^{|M|},1^{|\inp|},M(\inp),T^*),\ek\sample\rsetup(1^\secpar,1^\ell,\crs)]|\leq \negl(\secpar).
\end{align*}

Badrinarayanan et al.~\cite{AC:BFKSW19} gave a construction of strong output-compressing randomized encoding based on iO and the LWE assumption. 

\subsubsection{SNARK in the QROM}
Let $H:\bit^{2\secpar}\ra \bit^{\secpar}$ be a quantum random oracle.
A SNARK for an $\NP$ language $\lang$ associated with a relation $\rela$ in the QROM consists of two efficient oracle-aided classical algorithms $\pro_\snark^H$ and $\ver_\snark^H$.
\begin{description}
\item[$\pro_\snark^H$:] It is an instance $x$ and a witness $w$ as input and outputs a proof $\pi$.
\item[$\ver_\snark^H$:] It is an instance $x$ and a proof $\pi$ as input and outputs $\top$ indicating acceptance or $\bot$ indicating rejection.
\end{description}
We require SNARK to satisfy the following properties:

\noindent\textbf{Completeness.}
For any $(x,w)\in\rela$, we have 
\begin{align*}
    \Pr_{H}[\ver_\snark^H(x,\pi)=\top:\pi\sample \pro_\snark^H(x,w)]=1.
\end{align*}

\noindent\textbf{Extractability.}
There exists an efficient quantum extractor $\ext$ such that
for any $x$ and a malicious quantum prover $\tilde{\pro}_\snark^{H}$ making at most $q=\poly(\secpar)$ queries, if 
\begin{align*}
    \Pr_H[\ver_\snark^H(x,\pi):\pi \sample \tilde{\pro}_\snark^H(x)]
\end{align*}
is non-negligible in $\secpar$, then 
\begin{align*}
    \Pr_H[(x,w)\in \rela: w\sample \ext^{\tilde{\pro}_\snark}(x,1^q,1^\secpar)]
\end{align*}
is non-negligible in $\secpar$.\\

\noindent\textbf{Efficient Verification.}
If we can verify that $(x,w)\in \rela$ in classical time $T$, then for any $\pi \sample \tilde{\pro}_\snark^H(x)$, $\ver_\snark^H(x,\pi)$ runs in classical time $\poly(\secpar,|x|,\log T)$.

Chiesa et al.~\cite{TCC:ChiManSpo19} showed that there exists SNARK in the QROM that satisfies the above properties.
\section{Proof of Lemma~\ref{lem:Cauchy-Schwarz}}\label{app:proof_CS}
\begin{proof}
Since $M$ is a projective measurement, there exists a projection $\Pi$ such that 
\begin{align*}
    \Pr[M\circ \ket{\psi}=1]=\bra{\psi} \Pi \ket{\psi}.
\end{align*}
Then we have 
\begin{align*}
    \bra{\psi} \Pi \ket{\psi}&
    =\|\sum_{i=1}^{m} \Pi\ket{\psi_i}\|^2\\
    &\leq m \sum_{i=1}^{m}\|\Pi\ket{\psi_i}\|^2\\
    &=m\sum_{i=1}^{m}\bra{\psi_i} \Pi \ket{\psi_i}\\
    &=m \sum_{i=1}^{m}\|\ket{\psi_i}\|^2 \Pr[M\circ \frac{\ket{\psi_i}}{\|\ket{\psi_i}\|}=1]
\end{align*}
where we used the Cauchy-Schwarz inequality from the second to third lines.
\end{proof}
\section{Soundness of Two-Round Protocol}\label{app:fiat-shamir}

Here, we show that the soundness of the two-round protocol in Sec. \ref{sec:tworound} can be reduced to the  soundness of $m$-parallel version of the Mahadev's protocol.
For proving this, we borrow the following lemma shown in \cite{C:DFMS19}.


\begin{lemma}[{\cite[Theorem 2]{C:DFMS19}}]\label{lem:FS}
Let $\calY$ be finite non-empty sets. There exists a black-box polynomial-time two-stage quantum algorithm $\calS$ with the following property. Let $\A$ be an arbitrary oracle quantum algorithm that makes $q$ queries to a uniformly random $H:\calY\ra \bit^{m}$ and that outputs some $y\in \calY$ and output $a$. 
Then, the two-stage algorithm $\calS^{\A}$ outputs $y\in\calY$ in the first stage
and, upon a random $c\in \bit^{m}$ as input to the second stage, output $a$ so that for any $x_\circ\in \calX$ and any predicate $V$:
\begin{align*}
    \Pr_{c}\left[V(y,c,a):(y,a)\sample \langle\calS^{\A},c \rangle \right]\leq \frac{1}{O(q^2)}\Pr_{H}\left[V(y,H(y),a):(y,a)\sample \A^H\right]-\frac{1}{2^{m+1}q},
\end{align*}
where 
$(y,a)\sample \langle\calS^{\A},c \rangle$ means that $\calS^{\A}$ outputs $y$ and $a$ in the first and second stages respectively on the second stage input $c$.
\end{lemma}

We assume that there exists an efficient adversary $\A$ that breaks the soundness of the above two-round protocol.
We fix $x\notin \lang$ on which $\A$ succeeds in cheating.
We fix $(k,\td)$ that is in the support of the verifier's first message.
We apply Lemma~\ref{lem:FS} for $\A=\A(k)$ and $V=\ver_\out(k,\td,\cdot,\cdot,\cdot)$, to obtain an algorithm $\calS^{\A(k)}$ that satisfies
\begin{align*}
    &\Pr_{c}\left[V_{\out}(k,\td,y,c,a):(y,a)\sample \langle\calS^{\A(k)},c \rangle \right]\\
    \leq &\frac{1}{O(q^2)}\Pr_{H}\left[V_{\out}(k,\td,y,H(y),a):(y,a)\sample \A^H(k)\right]-\frac{1}{2^{m+1}q}.
\end{align*}
Averaging over all possible $(k,\td)$, we have
\begin{align*}
    &\Pr_{k,\td,c}\left[V_{\out}(k,\td,y,c,a):(y,a)\sample \langle\calS^{\A(k)},c \rangle \right]\\
    \leq &\frac{1}{O(q^2)}\Pr_{k,\td,H}\left[V_{\out}(k,\td,y,H(y),a):(y,a)\sample \A^H(k)\right]-\frac{1}{2^{m+1}q}.
\end{align*}
Since we assume that $\A$ breaks the soundness of the above two-round protocol,
\[
\Pr_{k,\td,H}\left[V_{\out}(k,\td,y,H(y),a):(y,a)\sample \A^H(k)\right]
\]
is non-negligible in $\secpar$.
Therefore, as long as $q=\poly(\secpar)$, 
\[
\Pr_{k,\td,c^*}\left[V_{\out}(k,\td,y,c^*,a):(y,a)\sample \langle\calS^{\A(k)},c^* \rangle \right]\]
is also non-negligible in $\secpar$.
Then, we construct an adversary $\B$ that breaks the soundness of parallel version of Mahadev's protocol as follows:
\begin{description}
\item[Second Message:] Given the first message $k$, $\B$ runs the first stage of $\calS^{\A(k)}$ to obtain $y$. It sends $y$ to the verifier.
\item[Forth Message:] Given the third message $c$, $\B$ gives $c$ to $\calS^{\A(k)}$ as the second stage input, and let $a$ be the output of it.
Then $\B$ sends $a$ to the verifier.
\end{description}
Clearly, the probability that $\B$ succeeds in cheating is 
\[
\Pr_{k,\td,c^*}\left[V_{\out}(k,\td,y,c^*,a):(y,a)\sample \langle\calS^{\A(k)},c^* \rangle \right],\]
which is non-negligible in $\secpar$.
This contradicts the soundness of $m$-parallel version of Mahadev's protocol (Theorem~\ref{thm:rep_soundness}).
Therefore we conclude that there does not exists an adversary that succeeds in the two-round protocol with non-negligible probability assuming LWE in the QROM.
\section{Omitted Proofs in Section \ref{sec:efficient}}
\subsection{Proof of Lemma~\ref{lem:adaptive_program}}\label{sec:proof_adaptive_program}
Here, we give a proof of Lemma~\ref{lem:adaptive_program}.
We note that the proof is essentially the same as the proof of \cite[Lemma 2.2]{EC:SaiXagYam18}.

Before proving the lemma, we introduce another lemma, which gives a lower bound for a decisional variant of Grover's search problem. 

\begin{lemma}[{\cite[Lemma C.1]{SonYun17}}] \label{lem:Decision_Grover}
Let $g_z:\bit^\ell \rightarrow \bit$ denotes a function defined as $g_{z}(z):=1$ and $g_z(z'):=0$ for all $z'\not=z$, and $g_{\bot}:\bit^\ell \rightarrow \bit$ denotes a function that returns $0$ for all inputs. Then for any quantum adversary $\B=(\B_1,\B_2)$ we have 
\[
\left|\Pr[1\sample\B_2(\ket{\st_{\B}},z) \mid \ket{\st_{\B}}\sample \B_1^{g_z}()]- \Pr[1\sample\B_2(\ket{\st_{\B}},z) \mid \ket{\st_{\B}}\sample \B_1^{g_{\bot}}()] \right|\leq q_1 \cdot2^{-\frac{\ell}{2}+1}.
\]
where $z\sample \bit^{\ell}$ and $q_1$ denotes the maximal number of queries by $\B_1$.
\end{lemma}

Then we prove Lemma~\ref{lem:adaptive_program}.

\begin{proof}(of Lemma~\ref{lem:adaptive_program}.)
We consider the following sequence of games.
We denote the event that $\game_i$ returns $1$ by $\TT_i$.
\begin{description}
\item[$\game_1$:] This game simulates the environment of the first term of LHS in the inequality in the lemma. Namely, the challenger chooses $z\sample \bit^{\ell}$, $H\sample \func(\bit^{\ell}\times \calX,\calY)$, $\A_1$ runs with oracle $H$ to generate $\ket{\st_{\A}}$, $\A_2$ runs on input $(\ket{\st_{\A}},z)$ with oracle $H$ to generate a bit $b$, and the game returns $b$.
\item[$\game_2$:] This game is identical to the previous game except that the oracle given to $\A_1$ is replaced with $H[z,G]$ where $G\sample \func(\calX,\calY)$.
\item[$\game_3$:] This game is identical to the previous game except that the oracle given to $\A_1$ is replaced with $H$ and the oracle given to $\A_2$ is replaced with $H[z,G]$.
We note that this game simulates the environment as in the second term of the LHS in the inequality in the lemma.
\end{description}
What we need to prove is $|\Pr[\TT_1]-\Pr[\TT_3]|\leq q_12^{-\frac{\ell}{2}+1}$.
First we observe that the change from $\game_2$ to $\game_3$ is just conceptual and nothing changes from the adversary's view since in both games, the oracles given to $\A_1$ and $\A_2$ are random oracles that agrees on any input $(z',x)$ such that $z'\neq z$ and independent on any input $(z,x)$.
Therefore we have $\Pr[\TT_2]=\Pr[\TT_3]$.
What is left is to prove $|\Pr[\TT_1]-\Pr[\TT_2]|\leq q_12^{-\frac{\ell}{2}+1}$.
For proving this, we construct an algorithm $\B=(\B_1,\B_2)$ that breaks Lemma~\ref{lem:adaptive_program} with the advantage $|\Pr[\TT_1]-\Pr[\TT_2]$ as follows:

\begin{description}
\item[$\B_1^{g^*}()$:] It generates $H\sample \func(\bit^{\ell}\times \calX,\calY)$ and  $G\sample \func(\calX,\calY)$, implements an oracle $O_1$ as
\begin{align*}
O_1(z',x)=
\begin{cases}
G(x) &\text{~if~}g^*(z')=1 \\
H(z',x) &\text{~else}
\end{cases},    
\end{align*}
runs $\ket{\st_{\A}}\sample\A_1^{O_1}()$ and outputs $\ket{\st_{\B}}\defeq \ket{\st_{\A}}$
\item[$\B_2(\ket{\st_{\B}}=\ket{\st_{\A}},z)$:]
It runs $b \sample \A_2^{H}(\ket{\st_{\B}},z)$ and outputs $b$.
\end{description}
It is easy to see that if $g^*=g_{\bot}$, then $\B$ perfectly simulates $\game_1$ for $\A$ and if $g^*=g_z$, then $\B$ perfectly simulates $\game_2$ for $\A$.
Therefore, we have $|\Pr[\TT_1]-\Pr[\TT_2]|\leq q_12^{-\frac{\ell}{2}+1}$ by Lemma~\ref{lem:adaptive_program}.
\end{proof} 

\subsection{Proof of Theorem \ref{thm:eff_soundness}}\label{app:proof_eff_soundness}
Here, we give a proof of Theorem \ref{thm:eff_soundness}.

\begin{proof}[of Theorem \ref{thm:eff_soundness}]
We fix $T$ and $x\notin \lang_T$.
Let $\A$ be a cheating prover.
First, we divides $\A$ into the first stage $\A_1$, which is given $(\crs_{\proeff},\crs_{\vereff})$ and the first message and outputs the second message $\ct'$ and its internal state $\ket{\st_{\A}}$, and the second stage $\A_2$, which is given the internal state $\ket{\st_{\A}}$ and the third message and outputs the fourth message $\pi_{\snark}$.
We consider the following sequence of games between an adversary $\A=(\A_1,\A_2)$ and a challenger. 
Let $q_1$ and $q_2$ be an upper bound of number of random oracle queries by $\A_1$ and $\A_2$, respectively.
We denote the event that the challenger returns $1$ in $\game_i$ by $\TT_i$.
\begin{description}
\item[$\game_1$:] This is the original soundness game.
Specifically, the game runs as follows:
\begin{enumerate}
    \item The challenger generates 
    $H\sample \func(\bit^{2\secpar}\times \bit^{2\secpar},\bit^{\secpar})$, 
    $\crs_\re\sample \bit^{\ell}$, 
    $s\sample \bit^{\ell_s}$, and $(\pk_{\fhe},\sk_{\fhe})\sample \fhekeygen(1^\secpar)$, and computes $\ek_\re\sample \rsetup(1^\secpar,1^\ell,\crs_\re)$,  
    $\ct\sample \fheenc(\pk_{\fhe},s)$, and  $\Menc\sample \renc(\ek_\re,M,s,T')$.
    \item $\A_1^{H}$ is given $\crs_{\proeff}\defeq \crs_\re$, $\crs_{\vereff}\defeq \ek_\re$ and the first message $(\Menc,\pk_{\fhe},\ct)$, and outputs the second message $\ct'$ and its internal state $\ket{\st_{\A}}$.
    \item The challenger randomly picks $z\sample \bit^{2\secpar}$.
    \item $\A_2^{H}$ is given the state $\ket{\st_{\A}}$ and the third message $z$ and outputs $\pi_{\snark}$.
    \item The challenger returns $1$ if  $\ver_{\snark}^{H(z,\cdot)}((x,\pk_{\fhe},\ct,\ct'),\pi_{\snark})=\top$ and $1 \la \fhedec(\sk_{\fhe},\ct')$ and $0$ otherwise.
\end{enumerate}

\item[$\game_2$:]
This game is identical to the previous game except that the oracles given to $\A_2$ and $V_{\snark}$ are replaced with $H[z,G]$ and $G$ in Step 4 and 5 respectively where $G\sample \func(\bit^{2\secpar},\bit^{\secpar})$ and $H[z,G]$ is as defined in Lemma~\ref{lem:adaptive_program}. 
We note that the oracle given to $\A_1$ in Step 2 is unchanged from $H$.

\item[$\game_3$:]
This game is identical to the previous game except that Step 4 and 5 are modified as follows:
\begin{description}
\item[4.] The challenger runs $e \sample \ext^{\A'_2[H,\ket{\st_{\A}},z]}((x,\pk_{\fhe},\ct,\ct'),1^{q_2},1^\secpar)$ where $\A'_2[H,\st_{\A},z]$ is an oracle-aided quantum algorithm that is given an oracle $G$ and emulates $\A_2^{H[z,G]}(\ket{\st_{\A}},z)$.
\item[5.] The challenger returns $1$ if 
$e$ is a valid witness for $(x,\pk_{\fhe},\ct,\ct')\in \lang_{\snark}$ and $1 \la \fhedec(\sk_{\fhe},\ct')$ and $0$ otherwise.
\end{description}


\item[$\game_4$:]
This game is identical to the previous game except that Step 5 is modified as follows:
\begin{description}
\item[5.] The challenger returns $1$ if 
$e$ is a valid witness for $(x,\pk_{\fhe},\ct,\ct')\in \lang_{\snark}$, and $\ver_\out(x,k,\td,e)=\top$ where $(k,\td)\sample \ver_1(1^\secpar,x;PRG(s))$ and $0$ otherwise.
\end{description}

\item[$\game_5$:]
This game is identical to the previous game except that $\ct$ is generated as $\ct\sample \fheenc(\pk_{\fhe},\allowbreak 0^{2\secpar})$ in Step 1.

\item[$\game_6$:]
This game is identical to the previous game except that $\crs_\re$, $\ek_\re$, and $\Menc$ are generated in a different way. Specifically, in Step $1$, the challenger computes
$(k,\td)\sample \ver_1(1^\secpar,x;PRG(s))$, $(\crs_\re,\Menc)\sample \calS_\re(1^\secpar,1^{|M|},1^{\ell_s},k,T^*)$, and $\ek_\re \sample \rsetup(1^\secpar,1^\ell,\crs_\re)$ where $T^*$ is the running time of $M(\inp)$.
We note that the same $(k,\td)$ generated in this step is also used in Step 5.

\item[$\game_7$:]
This game is identical to the previous game except that $PRG(s)$ used for generating $(k,\td)$ in Step 1 is replaced with a true randomness.
\end{description}

This completes the descriptions of games. 
Our goal is to prove $\Pr[\TT_1]=\negl(\secpar)$. We prove this by the following lemmas. 
Since Lemmas~\ref{lem:eff_gamehop_five}, \ref{lem:eff_gamehop_six}, and \ref{lem:eff_gamehop_seven} can be proven by straightforward reductions, we only give proofs for the rest of lemmas.

\begin{lemma}\label{lem:eff_gamehop_one}
We have $|\Pr[\TT_2]-\Pr[\TT_1]|\leq q_12^{-(\secpar+1)}$.
\end{lemma}
\begin{proof}
This lemma is obtained by applying Lemma~\ref{lem:adaptive_program} for $\B=(\B_1,\B_2)$ described below:
\begin{description}
\item[$\B_1^{O_1}$():]   It generates 
    $\crs_\re\sample \bit^{\ell}$,  $s\sample \bit^{\ell_s}$, and $(\pk_{\fhe},\sk_{\fhe})\sample \fhekeygen(1^\secpar)$, computes $\ek_\re\sample \rsetup(1^\secpar,1^\ell,\crs_\re)$,  $\ct\sample \fheenc(\pk_{\fhe},s)$, $\Menc\sample \renc(\ek_\re,M,s,T')$,  and $\ct\sample \fheenc(\pk_{\fhe},s)$, and sets $\crs_{\proeff}=\crs_\re$ and $\crs_{\vereff}\defeq \ek_\re$.
Then it runs $(\ct',\ket{\st_{\A}})\sample \A_1^{O_1}(\crs_{\proeff},\crs_{\vereff},x,(\Menc,\pk_{\fhe},\ct))$, and outputs $\ket{\st_\B}\defeq(\ket{\st_\A},x,\Menc,\ct,\ct',\sk_\fhe)$.\footnote{Classical strings are encoded as quantum states in a trivial manner.}
\item[$\B_2^{O_2}(\ket{\st_\B},z)$:] 
It runs $\pi_{\snark}\sample \A_2^{O_2}(\ket{\st_\A})$, and outputs $1$ if $\ver_{\snark}^{O_2(z,\cdot)}((x,\pk_{\fhe},\ct,\ct'),\pi_{\snark})=\top$ and $1 \la \fhedec(\sk_{\fhe},\ct')$ and $0$ otherwise.
\end{description}
\end{proof}

\begin{lemma}\label{lem:eff_gamehop_two}
If $\Pi_{\SNARK}$ satisfies the extractability and $\Pr[\TT_2]$ is non-negligible, then $\Pr[\TT_3]$ is also non-negligible.
\end{lemma}
\begin{proof}
Let $\transcript_3$ be the transcript of the protocol before the forth message is sent (i.e., $\transcript_3=(\crs_{\proeff},\crs_{\vereff},\Menc,\pk_{\fhe},\ct',z)$).
We say that $(H,\sk_\fhe,\transcript_3,\ket{\st_\A})$ is good if we randomly choose $G\sample \func(\bit^{2\secpar},\bit^{\secpar})$ and run $\pi_{\snark}\sample \A_2^{H[z,G]}(\ket{\st_\A})$ to complete the transcript, then the transcript is accepted (i.e., we have $\ver_{\snark}^{G}((x,\pk_{\fhe},\ct,\ct'),\pi_{\snark})=\top$ and $1 \la \fhedec(\sk_{\fhe},\ct')$) with non-negligible probability.
By a standard averaging argument, if $\Pr[\TT_2]$ is non-negligible, then a non-negligible fraction of $(H,\sk_\fhe,\transcript_3,\ket{\st_\A})$ is good when they are generated as in $\game_2$.
We fix good $(\transcript_3,\sk_\fhe,\ket{\st_\A})$.
Then by the extractability of $\Pi_{\SNARK}$, $\ext$ succeeds in extracting a witness for $(x,\pk_{\fhe},\ct,\ct')\in \lang_{\snark}$ with non-negligible probability. Moreover, since we assume $(H,\sk_\fhe,\transcript_3,\ket{\st_\A})$ is good, we always have $1 \la \fhedec(\sk_{\fhe},\ct')$ (since otherwise a transcript with prefix $\transcript_3$ cannot be accepted).
Therefore we can conclude that $\Pr[\TT_3]$ is non-negligible.
\end{proof}


\begin{lemma}\label{lem:eff_gamehop_four}
We have $\Pr[\TT_4]=\Pr[\TT_3]$.
\end{lemma}
\begin{proof}
If $e$ is a valid witness for $(x,\pk_{\fhe}, \ct,\ct')\in \lang_{\snark}$, then we especially have $\ct'= \fheeval(\pk_{\fhe},\allowbreak C[x,e],\ct)$.
By the correctness of $\Pi_\FHE$, we have $\fhedec(\sk_\fhe,\ct')=C[x,e](s)=(\ver_\out(x,k,\td,e)\overset{?}{=}\top)$ where $(k,\td)\sample \ver_1(1^\secpar,x;PRG(s))$.
Therefore, the challenger returns $1$ in $\game_4$ if and only if it returns $1$ in $\game_3$.
\end{proof}

\begin{lemma}\label{lem:eff_gamehop_five}
If $\Pi_{\FHE}$ is CPA-secure, then we have $|\Pr[\TT_5]-\Pr[\TT_4]|\leq \negl(\secpar)$.
\end{lemma}

\begin{lemma}\label{lem:eff_gamehop_six}
If $\Pi_{\RE}$ is secure, then we have
$|\Pr[\TT_6]-\Pr[\TT_5]|\leq \negl(\secpar)$.
\end{lemma}

\begin{lemma}\label{lem:eff_gamehop_seven}
If $\PRG$ is secure, then we have  $|\Pr[\TT_7]-\Pr[\TT_6]|\leq \negl(\secpar)$.
\end{lemma}

\begin{lemma}\label{lem:eff_gamehop_eight}
If $(\pro,\ver)$ satisfies soundness, then we have $\Pr[\TT_7]\leq \negl(\secpar)$.
\end{lemma}
\begin{proof}
Suppose that $\Pr[\TT_7]$ is non-negligible. Then we construct an adversary $\B$ against the underlying two-round protocol as follows:
\begin{description}
\item[$\B(k)$:] Given the first message $k$, it generates 
    $H\sample\func(\bit^{2\secpar}\times \bit^{2\secpar},\bit^{\secpar})$, $G\sample \func(\bit^{2\secpar},\allowbreak \bit^{\secpar})$, $z\sample \bit^{2\secpar}$,
    $(k,\td)\sample \ver_1(1^\secpar,x;PRG(s))$, $(\crs_\re,\Menc)\sample \calS_\re(1^\secpar,1^{|M|},1^{\ell_s},k,T^*)$, $\ek_\re \sample \rsetup(1^\secpar,1^\ell,\crs_\re)$,   and $(\pk_{\fhe},\sk_{\fhe})\sample \fhekeygen(1^\secpar)$,  computes  $\ct\sample \fheenc(\allowbreak \pk_{\fhe},0^{2\secpar})$,
    and sets $\crs_{\proeff}=\crs_\re$ and $\crs_{\vereff}\defeq \ek_\re$.
Then it runs $(\ct',\ket{\st_{\A}})\sample \A_1^{H}(\crs_{\proeff},\crs_{\vereff},\allowbreak x,(\Menc,\pk_{\fhe},\ct))$ and $e\sample \ext^{\A'_2[H,\ket{\st_{\A}},z]}((x,\pk_{\fhe},\ct,\ct'),1^{q_2},1^\secpar)$ and outputs $e$.
\end{description}
Then we can easily see that the probability that we have $\ver_\out(x,k,\td,e)$ is at least $\Pr[\TT_7]$.
Therefore, if the underlying two-round protocol is sound, then $\Pr[\TT_7]=\negl(\secpar)$.
\end{proof}

By combining Lemmas~\ref{lem:eff_gamehop_one} to~\ref{lem:eff_gamehop_seven}, we can see that if $\Pr[\TT_1]$ is non-negligible, then $\Pr[\TT_7]$ is also non-negligible, which contradicts Lemma~\ref{lem:eff_gamehop_eight}.
Therefore we conclude that $\Pr[\TT_1]=\negl(\secpar)$.
\end{proof}

\subsection{Two-Round CVQC Protocol with Efficient Verifier}\label{sec:omitted_two-round_eff}
Here, Here, we show that the number of rounds can be reduced to $2$ relying on another random oracle.
Namely, we observe that the third message of the scheme is just a public coin, and so we can apply the Fiat-Shamir transform similarly to Sec.\ref{sec:tworound}.
In the following, we describe the protocol for completeness.

Our two-round CVQC protocol $(\setupefffs,\proefffs,\verefffs=(\verefffsone,\verefffsout))$ for $\lang_T$ in the CRS+QRO model is described as follows.
Let $H:\bit^{2\secpar}\times \bit^{2\secpar}\ra \bit^{\secpar}$ be a quantum random oracle and $H':\bit^{\ell_{ct'}}\ra \bit^{2\secpar}$ be another quantum random oracle where $\ell_{\ct'}$ is the maximal length of $\ct'$ in the four-round scheme and $\ell$ and $T'$ be as defined in the previous section.

\begin{description}
\item[$\setupefffs(1^\secpar)$:]
The setup algorithm takes the security parameter $1^\secpar$ as input, generates $\crs_\re \sample \bit^{\ell}$ 
and computes $\ek_\re\sample \rsetup(1^\secpar,1^\ell,\crs_\re)$. 
Then it outputs a CRS for verifier $\crs_{\verefffs}\defeq \ek_\re$ and a CRS for prover $\crs_{\proefffs}\defeq \crs_\re$.
\item[$\verefffsone^{H,H'}$:] 
Given $\crs_{\verefffs}= \ek_\re$ and $x$, 
it generates $s\sample \bit^{\ell_s}$ and $(\pk_{\fhe},\sk_{\fhe})\sample \fhekeygen(1^\secpar)$,
computes $\ct\sample \fheenc(\pk_{\fhe},s)$
and  $\Menc\sample \renc(\ek_\re,M,s,T')$
where $M$ is a Turing machine that works as follows:
\begin{description}
\item[$M(s)$:] Given an input $s\in \bit^{\ell_s}$, it computes $(k,\td)\sample \ver_1(1^\secpar,x;PRG(s))$ and outputs $k$.
\end{description}
Then it sends $(\Menc,\pk_{\fhe},\ct)$ to $\proefffs$ and keeps $\sk_{\fhe}$ as its internal state.

\item[$\proefffstwo^{H,H'}$:] Given $\crs_{\proefffs}=\crs_\re$, $x$ and the message  $(\Menc,\pk_{\fhe},\ct)$ from the verifier, it computes $k\la \rdec(\crs_{\re},\Menc)$, $e\sample \pro_2(x,k)$, and $\ct'\la \fheeval(\pk_{\fhe},C[x,e],\ct)$ where $C[x,e]$ is a classical circuit defined above.
Then it computes $z\defeq H'(\ct')$, computes $\pi_\snark \sample \pro_{\snark}^{H(z,\cdot)}((x,\pk_{\fhe},\ct,\ct'),e)$
and sends $(\ct',\pi_\snark)$ to $\verefffs$.

\item[$\verefffsout^{H,H'}$:] 
It computes $z\defeq H'(\ct')$ and returns $\top$ if $\ver_{\snark}^{H(z,\cdot)}((x,\pk_{\fhe},\ct,\ct'),\pi_{\snark})=\top$ and $1 \la \fhedec(\sk_{\fhe},\ct')$ and $\bot$ otherwise.
\end{description}

\paragraph{Verification Efficiency.}
Clearly, the verification efficiency is preserved from the protocol in Sec.~\ref{sec:efficient-four}

\begin{theorem}[Completeness]\label{thm:efffs_completeness}
For any $x\in \lang_T$, 
\begin{align*}
\Pr\left[\langle\proefffs^{H,H'}(\crs_{\proefffs}),\verefffs^{H,H'}(\crs_{\verefffs})\rangle(x)=\bot \right]=\negl(\secpar)    
\end{align*}
where $(\crs_{\proefffs},\crs_{\verefffs})\sample \setupefffs(1^\secpar)$.
\end{theorem}

\begin{theorem}[Soundness]\label{thm:efffs_soundness}
For any $x\notin \lang_T$ any efficient quantum cheating prover $\A$, 
\begin{align*}
\Pr\left[\langle\A^{H,H'}(\crs_{\proefffs},\crs_{\vereff}),\verefffs^{H,H'}(\crs_{\verefffs})\rangle(x)=\top\right]=\negl(\secpar)    
\end{align*}
where $(\crs_{\proefffs},\crs_{\verefffs})\sample \setupefffs(1^\secpar)$.
\end{theorem}
This can be reduced to Theorem~\ref{thm:eff_soundness} similarly to the proof of soundness of the protocol in Sec.~\ref{sec:tworound}.
This completes the proof of Theorem \ref{thm:Eff}

\fi
\end{document}